\documentclass[a4paper,UKenglish,cleveref, autoref, thm-restate]{lipics-v2021}

\usepackage{blindtext}
\usepackage{hyperref}
\usepackage{graphicx}
\usepackage{amsfonts} 
\usepackage{todonotes}
\usepackage{algpseudocode}
\usepackage{algorithm}
\allowdisplaybreaks

\usepackage{pifont}
\usepackage{mathtools}
 \usepackage{pdfpages}

\usepackage{subcaption}

\def\A{{\cal A}}

\def\S{{\mathcal S}}
\def\T{{\mathcal T}}
\def\OO{{\cal O}} 

\usepackage{mathtools}

\nolinenumbers 

\DeclarePairedDelimiter\abs{\lvert}{\rvert}%
\DeclarePairedDelimiter\norm{\lVert}{\rVert}%

\makeatletter
\let\oldabs\abs
\def\abs{\@ifstar{\oldabs}{\oldabs*}}
\let\oldnorm\norm
\def\norm{\@ifstar{\oldnorm}{\oldnorm*}}

\usepackage{bm}
 \newtheorem{property}{Property}

\newcommand{\IR}{\mathbb{R}}
\newcommand{\IZ}{\mathbb{Z}}
\newcommand{\Z}{\mathbb{Z}}

\DeclareMathOperator{\Cell}{Cell}

\usepackage{color}

\DeclareMathOperator{\LIR}{\textsc{LayeredRIR}}
\DeclareMathOperator{\ANC}{\textsc{CloseToCenter}}
\DeclareMathOperator{\RIR}{\textsc{RIR}}
\DeclareMathOperator{\ES}{\textsf{\textsc{Algo-ES}}}
\DeclareMathOperator{\HHR}{\textsf{\textsc{LayeredES}}}

\def\OO{{\cal O}} 
\def\Q{{\cal Q}}
\def\A{{\cal A}}
\def\H{{\mathcal H}}
\DeclareMathOperator{\interior}{int}

\def\I{{\mathcal I}}
\def\K{{\mathcal K}}
\def\P1{{\mathcal P_1}}
\def\P2{{\mathcal P_2}}
\def\P{{\mathcal P}}
\def\B{{\mathcal B}}

\title{{New Lower Bound and Algorithms for  Online Geometric Hitting Set Problem}}
\titlerunning{New Lower Bound and Algorithms for Online Geometric Hitting Set Problem}

 \author{Minati De}{Dept. of Mathematics, Indian Institute of Technology Delhi, New Delhi, India}{minati@maths.iitd.ac.in}{https://orcid.org/0000-0002-1859-8800}{Partially supported by SERB MATRICS Grant MTR/2021/000584.}

 \author{Ratnadip Mandal}{Dept. of Mathematics, Indian Institute of Technology Delhi, New Delhi, India}{maz218522@iitd.ac.in}{[orcid]}{Supported by CSIR, India, File Number- 09/0086(13712)/2022-EMR-I.}

 \author{Satyam Singh}{Dept. of Mathematics, Indian Institute of Technology Delhi, New Delhi, India}{satyam.singh@maths.iitd.ac.in}{[orcid]}{Supported by CSIR, India, File Number-09/086(1429)/2019-EMR-I.}

 \authorrunning{M. De, R. Mandal and S. Singh} 

\Copyright{Minati De, Ratnadip Mandal and Satyam Singh} 



\category{} 

\relatedversion{} 

\EventEditors{John Q. Open and Joan R. Access}
\EventNoEds{2}
\EventLongTitle{42nd Conference on Very Important Topics (CVIT 2016)}
\EventShortTitle{CVIT 2016}
\EventAcronym{CVIT}
\EventYear{2016}
\EventDate{December 24--27, 2016}
\EventLocation{Little Whinging, United Kingdom}
\EventLogo{}
\SeriesVolume{42}
\ArticleNo{23}

\begin{document}

\maketitle
\begin{abstract}

The hitting set problem is one of the fundamental problems in combinatorial optimization and is well-studied in offline setup. We consider the online hitting set problem, where only the set of points is known in advance, and objects are introduced one by one. Our objective is to maintain a minimum-sized hitting set by making irrevocable decisions. Here, we present the study of two variants of the online hitting set problem depending on the point set. In the first variant, we consider the point set to be the entire $\mathbb{Z}^d$, while in the second variant, we consider the point set to be a finite subset of $\mathbb{R}^2$.

If you use points in $\mathbb{Z}^d$ to hit homothetic hypercubes in $\mathbb{R}^d$ with side lengths in $[1,M]$, we show that the competitive ratio of any algorithm is $\Omega(d\log M)$, whether it is deterministic or random. This improves the recently known deterministic lower bound of $\Omega(\log M)$ by a factor of $d$. Then, we present an almost tight randomized algorithm with a competitive ratio $O(d^2\log M)$ that significantly improves the best-known competitive ratio of $25^d\log M$. Next, we propose a simple deterministic ${\lfloor\frac{2}{\alpha}+2\rfloor^d}$$\left(\lfloor\log_{2}M\rfloor+1\right)$ competitive  algorithm to hit similarly sized {$\alpha$-fat objects} in $\mathbb{R}^d$ having diameters in the range $[1, M]$ using points in $\mathbb{Z}^d$.
This improves the current best-known upper bound by a factor of at least $5^d$.

Finally, we consider the hitting set problem when the point set consists of $n$ points in $\mathbb{R}^2$, and the objects are homothetic regular $k$-gons having diameter in the range $[1, M]$. We present an $O(\log n\log M)$ competitive randomized algorithm for that. Whereas no result was known even for squares.

In particular, our results answer some of the open questions raised by Khan et al. (\textit{Symposium on Computational Geometry}, 2023) and Alefkhani et al. (\textit{Approximation and Online Algorithm}, 2023). Some of the techniques used in solving them are simple yet nontrivial, and that is the strength of the paper.

\end{abstract}
\newpage
\pagenumbering{arabic}
\section{Introduction}
The hitting set problem is one of the most fundamental problems in combinatorial optimization.
Let $\Sigma=(\P,\cal S)$ be a \emph{range space}, where  $\cal P$ is a set of \emph{elements} and $\cal S$ is a collection of subsets of $\P$ known as \emph{ranges}. A subset $\cal H \subseteq \cal P$ is called a \emph{hitting set} of the range space $\Sigma=(\P,\cal S)$ if the set $\cal H$ intersects every range $s$ in $\cal S$. 
The objective of the \emph{hitting set problem} is to find a hitting set with the minimum cardinality.
This problem is one of Karp’s 21 classic NP-hard problems~\cite{10.6/574848}, and the best approximation factor achievable in polynomial time is $\Theta(\log |\P|)$ assuming $\text{P}\neq \text{NP}$~\cite{DinurS14,Feige98}.

Due to several applications in  VLSI design, database management, resource allocation, facility location and wireless networks, researchers have considered the hitting set problem in geometric setup (also known as \emph{geometric hitting set problem}), where the set $\cal P$ is a collection of points from $\IR^d$, and the set $\cal S$ is a finite family of geometric objects chosen from some infinite class (hypercubes, balls, etc.)~\cite{AgarwalP20, ChanH20, Ganjugunte11, MustafaR10}. 
By a slight misuse of the notation, we will use $\cal S$ to signify both the set of ranges as well as the set of objects that define these ranges.
The hitting set problem remains NP-hard, even for simple objects like unit disks and unit squares~\cite{FowlerPT81}. A PTAS is known when the geometric objects are half-spaces in $\IR^3$, and when objects are  $r$-admissible  regions in a plane (this includes pseudo-disks)~\cite{MustafaR10}.



Alon et al.~\cite{AlonAABN09} initiated the study of the online hitting set problem in a setup where the whole range space $\Sigma=(\P,{\cal S})$ is known in advance, but the order of arrival of the elements of ${\cal S}$ is not known {a priori}. In this setup, they proposed a deterministic $O(\log |\P| \log |\cal S|)$ competitive algorithm and obtained a lower bound of {$\Omega\left(\frac{\log |\P| \log |\cal S|}{\log\log |\P| +\log\log |\cal S|}\right)$}.
In the same setup,  Even and Smorodinsky~\cite{EvenS14} studied the problem for various geometric range spaces and achieved an optimal competitive ratio of $\Theta(\log |\P|)$ when objects are intervals, half-planes, and unit disks, and the set $\P$ is a finite subset of points in $\IR^2$. In this paper, we ask whether it is possible to obtain a similar result if we don't know the set $\cal S$ in advance, and we answer affirmatively when $\cal S$ is similarly sized homothetic copies of any regular $k$-gon.



Motivated by various applications, De and Singh~\cite{DeS24b} considered the online hitting set problem, where ${\cal P}={\mathbb Z}^d$ and $\cal S$ consists of unit balls or hypercubes in $\IR^d$. For unit hypercubes, they proved that every deterministic algorithm has a competitive ratio of at least~$d+1$; whereas they propose a randomized algorithm achieving a competitive ratio of~$O(d^2)$ for $d\geq3$. For unit balls in $\IR^d$, they proposed a deterministic algorithm having a competitive ratio of $O(d^4)$.

Khan et al.~\cite{KhanLRSW23} studied the problem, where, for a given $M\in\IR^{+}$, $\P$ is a finite subset of $[0,M)^2\cap\mathbb{Z}^2$ and the set $\cal S$ consists of axis-aligned integer squares $s\subseteq [0, M)^2$ in $\IR^2$. Here, an integer square is a square whose all corners have integral coordinates. They obtained an $O(\log M)$ competitive algorithm for this variant. This result is tight since the lower bound is $\Omega(\log M)$ which even holds for intervals~\cite{EvenS14}. Observe that, since the integer squares are from $[0, M)^2$, the side length of the squares is in the range $[1, M]$.
In the same setup, no result is known when the squares are not integer. In fact, the authors~\cite{KhanLRSW23} posed this as an open problem. Motivated by this, in this paper, we have considered a more general range space where ${\cal P}$ is a finite subset of $\IR^2$
 and the set $\cal S$ consists of squares (not necessarily integer squares) of side length in the range $[1,M]$.
%

When $\P=(0, M)^d\cap\mathbb{Z}^d$ and ${\cal S}$ is a family of $\alpha$-fat objects\footnote{Definition of fat object is given in Section~\ref{notation}. Note that this definition is slightly different than~\cite{AlefkhaniKM23}. If $\hat{\alpha}$ and $\alpha$ are the fatness of an object according to the definition of~\cite{AlefkhaniKM23} and ours, respectively,  then $\hat{\alpha}\leq \frac{1}{\alpha}$.} in $(0, M)^d$, very recently, Alefkhani et al.~\cite{AlefkhaniKM23} proposed a deterministic online algorithm with a competitive ratio of $(\frac{4}{\alpha} +1)^{2d} \log M$. They proved that every deterministic algorithm has a competitive ratio of~$\Omega(\log M)$ for this problem. They~\cite{AlefkhaniKM23} noted the vast gap between the lower and the upper bounds, and proposed to bridge this gap as an open problem. In this paper, we propose an exponential factor improvement in the upper bound and give an improved randomized lower bound.



\subsection{Our Contributions and Comparison to Prior Work}
We study the online hitting set problem when only the set $\P$ of points is known in advance.
The set $\cal S$ of objects is not known a priori; elements of this set are revealed one by one. The algorithm needs to maintain a hitting set for the revealed objects by choosing points from $\P$.
Once a point is included in the solution set, the decision cannot be reverted in the future.
Below, we mention each contribution formally and discuss the context and comparison to prior work. 
   
    \textbf{A New Improved Lower Bound.} First,  we establish that the competitive ratio of any algorithm, whether deterministic or randomized, is $\Omega(d\log M)$ for hitting homothetic hypercubes with side lengths in the range $[1, M]$ using points in $\mathbb{Z}^d$ (Theorem~\ref{thm: random lower bound homothetic hypercubes}). This improves the current best-known deterministic lower bound of $\Omega(\log M)$~\cite{AlefkhaniKM23} by a factor of $d$.
    %
    %
    To obtain the lower bound, we use Yao's minimax principle (as in~\cite{ChanZ09}). We first construct a complete binary tree $\T$ of height $O(d \log M)$, where each path from the root to a leaf of $\T$ represents a sequence of similarly sized homothetic hypercubes. If any path of $\T$ from the root node to one of the leaf nodes is presented at random as an input sequence,  then we bound the competitive ratio of any deterministic algorithm for this random input sequence.

    \textbf{An Almost Tight Algorithm for Hypercubes.} Next, we consider the hitting set problem using points in $\mathbb{Z}^d$, where objects are (axis-aligned) hypercubes having side lengths in $[1, M]$ and  $d\geq3$.
  Since the value of $\alpha$ is $1$ for (axis-aligned) hypercubes, the best known competitive ratio for this problem was~$25^d\log M$ due to Alefkhani et al.~\cite{AlefkhaniKM23}.
We propose a randomized algorithm $\LIR$ that achieves a competitive ratio of~$O(d^2\log M)$ (Theorem~\ref{thm: hyp_ub}).
The algorithm $\LIR$ uses a generalized version of an algorithm by De and Singh~\cite{DeS24b}, designed for hitting unit hypercubes using points in $\mathbb{Z}^d$. To accomplish this, we introduce two concepts: \emph{layer formation} and \emph{core construction}. Although the layer formation technique has been used before~\cite{CaragiannisFKP07}, our approach differs from prior implementations.
We partition the set of homothetic hypercubes into layers such that each hypercube within a particular layer contains at least $2^d$ points and at most $3^d$ points from $(s\mathbb{Z})^d$, where the value of $s$ will depend on the layer of the object. Since hypercubes within each layer are not of the same size, we cannot directly apply the generalized version of De and Singh's algorithm. This takes us to the introduction of the core construction concept, where we create a core object for each hypercube in a layer. This core object contains the same set of points from $(s\mathbb{Z})^d$ as the original object but has a fixed side length dependent on the layer. This allows us to use the generalized version of De and Singh's algorithm effectively.
 Here is a brief overview of our algorithm. Upon the arrival of a new hypercube, if it is not hit by previously selected points, we first determine the layer. We then evoke the generalized algorithm of~\cite{DeS24b} to hit the core of the object using points in $(s\mathbb{Z})^d$, where $s$ will depend on the layer. We show that algorithm $\LIR$ achieves a competitive ratio of~$O(d^2\log M)$ (Theorem~\ref{thm: hyp_ub}).

   \textbf{An Exponential Improved Competitive Algorithm for Fat Objects.} 
   Fat objects have several desirable properties and  are well-studied due to their various applications and appearance in geometric problems~\cite{Chan03, 
 BergSVK02, EfratKNS00}.
We focus on the hitting set problem, where the set $\P=\IZ^d$ and the set $\cal S$ consists of similarly sized $\alpha$-fat objects in $\IR^d$ having widths in the range $[1, M]$. For this, we present a deterministic algorithm $\ANC$, which is as follows: Upon the arrival of a new object $\sigma$ with width $w$, if $\sigma$ is not hit by the previously selected points, our algorithm identifies a point $r\in (2^{i+1}\mathbb{Z})^d$ to hit $\sigma$, where $i=\lfloor\log w\rfloor$. We show that this algorithm achieves a competitive ratio of at most~${\lfloor\frac{2}{\alpha}+2\rfloor^d} \left(\lfloor\log_{2}M\rfloor+1\right)$ (Theorem~\ref{thm:fat_ub}). 
The previous best-known upper bound for this problem due to Alefkhani et al.~\cite{AlefkhaniKM23} was $\left(\frac{4}{\alpha}+1\right)^{2d}\log M$. Note that $\left(\frac{4}{\alpha}+1\right)^{d}\left(\frac{4}{\alpha}+1\right)^{d}\geq \left(5\right)^d\left(\frac{4}{\alpha}+1\right)^{d}>\left(5\right)^d\left(\frac{2}{\alpha}+2\right)^{d}$.  Since the value of $\alpha$ is in the range $(0,1]$ and $\frac{4}{\alpha}+1=\frac{2}{\alpha}+\frac{2}{\alpha}+1\geq \frac{2}{\alpha}+2+1>\frac{2}{\alpha}+2$, the first and second inequalities hold.
Thus, our result improves the upper bound of Alefkhani et al.~\cite{AlefkhaniKM23}  by a factor of at least $5^d$ which is an exponential improvement.


    \textbf{An $O(\log n)$ Competetive Algorithm for Similarly Sized Homothetic Regular $k$-gons.}
     Finally, we consider the hitting set problem, where the set $\P\subset\IR^2$ consists of $n$ points and the set $\cal S$ consists of similarly sized homothetic regular $k$-gons ($k\geq 4$) having diameters in the range $[1, M]$.
       For this, we 
     propose a randomized algorithm $\HHR$
     having a  competitive ratio of  $O(\log n\log M)$ (Theorem~\ref{Theo:Main_Bounded}). In this setup, no algorithm was known, even for  squares, as mentioned earlier. In fact, 
     both Alefkhani et al.~\cite{AlefkhaniKM23} and Khan et al.~\cite{KhanLRSW23} posed it as an open problem to design an online algorithm having a competitive ratio of $O(\log n)$ for hitting squares using a set $P\subset \mathbb{R}^2$  of $n$ points. Thus, for a fixed $M$, the above-mentioned outcome affirmatively answers this open problem. On the way to design our algorithm $\HHR$, we propose an $O(\log n)$ competitive deterministic algorithm for translates of a regular $k$-gon (Theorem~\ref{thm:translates}) that is similar to the algorithm of  Even and Smorodinsky~\cite{EvenS14} for hitting translates of a disk. We show how to implement their algorithm without the knowledge of the set of objects beforehand. In particular, we show how to find the extreme points (defined in Section~\ref{sec: homothetic regular k-gons}) used to hit the introduced objects without knowing the objects (Lemma~\ref{clm:ext}). Along the way of extending the result of unit balls to regular $k$-gons, we obtain a key property~(Lemma~\ref{lemma:regular}) of regular $k$-gons that apparently was known only for unit disks. We refer to the property as \emph{angle property}. We feel this would be of independent interest, and it will find its applications in other geometric problems.




\subsection{Other Related Works}
The continuous version of the hitting set problem, where the point set $\P=\mathbb{R}^d$, is known as the \emph{piercing set problem}. When the set $\S$ consists of translated copies of a convex object, the corresponding piercing set problem is equivalent to the \emph{unit covering problem}~\cite{DeJKS24}. Charikar et al.~\cite{CharikarCFM04} considered the online unit covering problem and presented an $O(2^dd\log d)$ competitive algorithm for balls in $\IR^d$, and they also gave a deterministic lower bound of $\Omega\left(\frac{\log d}{\log\log\log d}\right)$ for this problem. Later, Dumitrescu et al.~\cite{DumitrescuGT20} improved the upper and lower bounds of the competitive ratio to $O({1.321}^d)$ and $\Omega(d+1)$, respectively. In particular, they~\cite{DumitrescuGT20} obtained competitive ratios of 5 and 12 for balls in $\IR^2$ and $\IR^3$, respectively. 
Dumitrescu and Tóth~\cite{DumitrescuT22} proved that the lower bound of the unit covering problem for hypercubes in $\mathbb{R}^d$ is at least $2^d$, and also this ratio is optimal, as it is obtainable through a simple deterministic algorithm that allocates points to a predefined set of hypercubes~\cite{ChanZ09}. Recently, De et al.~\cite{DeJKS24} studied the online piercing set problem for $\alpha$-fat objects in $\mathbb{R}^d$, having widths in the range $[1, M]$. They~\cite{DeJKS24} presented a deterministic algorithm for this problem with a competitive ratio of at most $\left(\frac{2}{\alpha}\right)^d((1+\alpha)^d-1)\lceil\log_{(1+\alpha)} \left(\frac{2M}{\alpha}\right)\rceil +1$. 
We refer to~\cite{DeJKS24} for other results regarding the online piercing set problem of various objects.

\section{Notation and Preliminaries}\label{notation}
We use $\mathbb{Z}^{+}$ and $\IR^{+}$ to denote the set of positive integers and positive real numbers, respectively. We use $[n]$ to represent the set  $\{1,2,\ldots,n\}$, where $n\in\mathbb{Z}^{+}$.
For any $\beta\in\mathbb{R}$, we use $\beta\mathbb{Z}$ to denote the set $\left\{\beta z \ |\ z\in\mathbb{Z}\right\}$, where $\mathbb Z$ is the set of integers.
\begin{observation}\label{obs:points}
    Let $\sigma$ be a hypercube with side lengths between $\ell\beta$ and $r\beta$, where $\ell$, $r$, $\beta\in \IR^{+}$ and $\ell<r$. Then, the hypercube $\sigma$ contains at least $\lfloor\ell\rfloor^d$ and at most $\lfloor r+1\rfloor^d$ points from $(\beta\IZ)^d$. 
\end{observation}

For any point $p\in\IR^d$, we use $p(x_i)$ to denote the $i$th coordinate of $p$, where $i\in[d]$. The point $p$ is an \emph{integer point} if for each $i\in[d]$, the coordinate $p(x_i)$ is an integer. The \emph{convex hull} of a set $P$ of points is defined as the intersection of all convex regions containing $P$.

By an \emph{object}, we refer to a compact set in $\mathbb{R}^d$ having a non-empty interior. For any object~$\sigma$, we use $\partial\sigma$ to denote the \emph{boundary} of the object, and $\text{int}(\sigma)=\{x\in \sigma\ |\ x\notin \partial\sigma\}$ refers the \emph{interior} of the object. A set $\cal S$ of objects is said to be \emph{similarly sized} if the ratio of the largest diameter of an object in $\cal S$ to the smallest diameter of an object in $\cal S$  is bounded by a fixed constant.
 A set of the form $\sigma+{\bf v}=\{c+{\bf v}\ |\ c\in \sigma\}$, where ${\bf v}\in\IR^d$, is called a \emph{translated copy/translate} of $\sigma$. 
A set of the form $\lambda \sigma+{\bf v}=\{\lambda c+{\bf v}\ |\ c\in \sigma\}$, where ${\bf v}\in\IR^d$ and $\lambda\in \mathbb{R}^+$, is called a \emph{homothetic copy} of $\sigma$. 
A set $\cal S$ of objects is said to be \emph{homothetic} if each object $\sigma\in \cal S$ is a homothetic copy of every other object $\sigma' \in \cal S$. {A set $\cal S$ of objects is said to be \emph{similarly sized homothetic} if $\cal S$ is similarly sized as well as homothetic.}

 A number of different definitions of fatness (not extremely long and skinny) are available in the geometry literature~\cite{AlefkhaniKM23, Chan03, DeJKS24}. 
We use a definition similar to~\cite{DeJKS24}. 
Let $\sigma$ be an {object}. For any point $x\in \sigma$, we define $\alpha(x)= \frac{\min_{y\in \partial\sigma}d_{\infty}(x,y)}{\max_{y\in \partial\sigma}d_{\infty}(x,y)}$, where $d_{\infty}(.,.)$  represents the distance under the $L_{\infty}$-norm.
The \emph{aspect ratio} $\alpha(\sigma)$ of the object $\sigma$ is defined as the maximum value  of $\alpha(x)$ for any point  $x\in \sigma$, i.e., $\alpha(\sigma)=\max \{\alpha(x)  : x \in  \sigma\}$. An object is considered an \emph{$\alpha$-fat object} if its aspect ratio is at least $\alpha$. 
We say that a set $\cal S$  of objects is \emph{fat} if there exists a constant $0<\alpha\leq 1$ such that each object in $\cal S$ is $\alpha$-fat. 
Note that, for a set $\cal S$ of fat objects,  each object $\sigma\in {\cal S}$ does not need to be convex and also connected.
For an $\alpha$-fat object, the value of $\alpha$ is invariant under translation, reflection, and scaling.
 We refer~\cite{DeJKS24} for various properties of $\alpha$-fat objects.


A point $c\in \sigma$ with $\alpha(c)=\alpha(\sigma)$ is defined as a \emph{center} of the object $\sigma$. Note that the center of an $\alpha$-fat object might not be unique.  If an object $\sigma$ has multiple points satisfying the property of a center, then for our purpose, we can arbitrarily choose any one of them as the center of $\sigma$. 
The minimum (respectively, maximum) {$L_{\infty}$ distance} from the center to the boundary of the object is referred to as the \emph{width} (respectively, \emph{height}) of the object. 
For a fixed $\alpha$, a set $\cal S$ of $\alpha$-fat objects is said to be \emph{similarly sized fat objects} when the ratio of the largest width of an object in $\cal S$ to the smallest width of an object in $\cal S$  is bounded by a fixed constant.

\section{Lower Bound for Similarly Sized Homothetic Hypercubes in $\IR^d$ Using Points in $\mathbb{Z}^d$}\label{sec: random lower bound}
In this section,  for hitting similarly sized homothetic hypercubes having widths in the range $[1, M]$ using points in $\IZ^d$, we prove that the lower bound of the competitive ratio of any randomized algorithm is  $\Omega(d\log M)$.  As mentioned earlier, throughout this section, we use \emph{width} of a hypercube to denote half of its side length.

    For simplicity of expressions, we assume that $M$ is an integral power of $4$. 
    Additionally, we assume that a hypercube $\sigma$ is hit by some point $p$ if $p$ lies in the interior of $\sigma$. We can remove this assumption by replacing each constructed hypercube with an infinitesimally shrunken version of itself.

    We use Yao's minimax principle (as in~\cite{ChanZ09}) to prove the lower bound. 
    To establish it, we create a complete binary tree $\T$ of height\footnote{The height of a tree is the number of vertices lying in the path from the root node to a leaf node.} $O(d \log M)$, where any path from the root to a leaf represents a sequence of hypercubes having widths in the range $[1, M]$, and the common intersection region of these hypercubes contains an integer point. 
    Let $\cal S$ be a set consisting of all sequences corresponding to {$2^{O(d \log M)}$} many paths from the root to each leaf node in $\T$. 
    We show that if any of the sequences of $\cal S$ is chosen randomly with uniform distribution, then the expected number of points placed by any deterministic algorithm is $\Omega(d \log M)$, whereas an offline optimum algorithm needs only one point. Consequently, the expected competitive ratio on any sequence from $\cal S$ is $\Omega(d\log M)$.
    A sequence in $\cal S$ can be chosen uniformly at random as follows. The root is introduced with certainty. Then, at each internal node, with equal probability, either the left or the right child is chosen.

    \begin{figure}[htbp]
    \centering
    \begin{subfigure}[b]{0.49\textwidth}
        \centering
        \includegraphics[page=3, width= 60 mm]{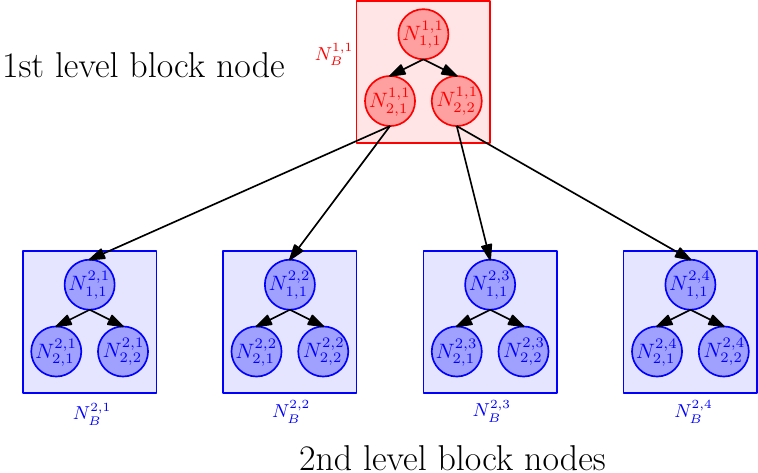}
    \end{subfigure}
    \hfill
    \begin{subfigure}[b]{0.49\textwidth}
        \centering
        \includegraphics[page=4, width=40 mm]{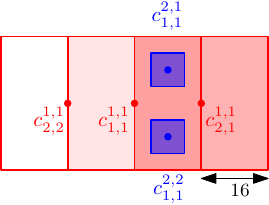}
    \end{subfigure}
    \caption{Construction of the tree $\T$ and $\T_B$ when $d=2$ and $M=16$. The red and blue colors represent the $1$st and $2$nd level block nodes, respectively. Here, $c^{1,1}_{1,1}=(0, 0)$, $c^{1,1}_{2,1}=(16, 0)$, $c^{2,1}_{1,1}=(8, 8)$ and so on.}
    \label{fig:lower homothetic hypercube}
    \end{figure}

    Now, we describe the construction of the tree $\T$. The tree $\T$ can be visualized as a complete $2^d$-ary tree $\T_B$ of height $\log_4 M$ such that each node of $\T_B$, called \emph{block node}, itself represents a complete binary tree of height $d$. See Figure~\ref{fig:lower homothetic hypercube} for an example. For each $k\in[\log_4 M]$ and $l\in[2^{d(k-1)}]$, let $N^{k,l}_B$ denote the $l$th block node at the $k$th level of $\T_B$, where at each level, we number the block nodes from left to right. So, the $2^d$ children block nodes of $N^{k,l}_B$ are $N^{k+1,2^d(l-1)+1}_B$, $N^{k+1,2^d(l-1)+2}_B,$ $\dots, N^{k+1,2^d(l-1)+2^d}_B$. Let  $\T^{k,l}_B$ be a complete binary tree corresponding to a block node $N^{k,l}_B$. For each $i\in[d]$ and $j\in[2^{i-1}]$, let $N^{k,l}_{i,j}$ denote the $j$th node at the $i$th level of the tree $\T^{k,l}_B$, where at each level, we number the nodes from left to right. Now, for each $i\in [d-1]$ and $j\in[2^{i-1}]$, two children nodes of an internal node $N^{k,l}_{i,j}$ of $\T^{k,l}_B$ are $N^{k,l}_{i+1,2j-1}$ and $N^{k,l}_{i+1,2j}$, respectively. Next, for $k<\log_4 M$, we describe the parent-child relation between a leaf of a tree $\T^{k,l}_B$ with the root of a tree corresponding to the child block node of $N^{k,l}_B$. 
    For a leaf node $N^{k,l}_{d,j}$ of $\T^{k,l}_B$, where $j\in [2^{d-1}]$, two children are $N^{k+1,2^d(l-1)+2j-1}_{1,1}$ and $N^{k+1,2^d(l-1)+2j}_{1,1}$, respectively.
   Thus, $\T$ is a complete binary tree consisting of nodes $N^{k,l}_{i,j}$, where $k\in[\log_4 M]$, $l\in[2^{d(k-1)}]$, $i\in[d]$ and $j\in[2^{i-1}]$, such that $N^{1,1}_{1,1}$ is the root of it. As a result, the height $h$ of $\T$ is $ h  = d\log_4 M= \frac{d}{2}\log_2 M$.

    Corresponding to each node $N_{i,j}^{k,l}$, we have a hypercube having a {width} $\frac{M}{4^{k-1}}$ centered at a point $c^{k,l}_{i,j}\in \IR^d$. Now, we explain the construction of the centers of all the hypercubes. This will be done recursively.
    \begin{itemize}
\item  First, we assume that $c^{1,1}_{1,1}$ is at the origin.

\item Now, consider an internal node $N^{k,l}_{i,j}$ of $\T^{k,l}_B$, where $k\in[\log_4 M]$, $l\in[2^{d(k-1)}]$, $i\in[d-1]$ and $j\in[2^{i-1}]$. Recall that the two children nodes of $N^{k,l}_{i,j}$ are $N^{k,l}_{i+1,2j-1}$ and $N^{k,l}_{i+1,2j}$, respectively. We define $c^{k,l}_{i+1,2j-1}=c^{k,l}_{i,j}+\frac{M}{4^{k-1}}{\bf e}_{i}$ and $c^{k,l}_{i+1,2j}=c^{k,l}_{i,j}-\frac{M}{4^{k-1}}{\bf e}_{i}$, where ${\bf e}_i$ is the $i$th standard unit vector in $\IR^d$. 
    In other words, the center $c^{k,l}_{i+1,2j-1}$ (respectively, $c^{k,l}_{i+1,2j}$) corresponding to the left (respectively, right) child of $N^{k,l}_{i,j}$ is a translated copy of $c^{k,l}_{i,j}$, differing only in the $i$th coordinate. In particular, for the left (respectively, right) child, this translation is by $\frac{M}{4^{k-1}}{\bf e}_{i}$ (respectively, $-\frac{M}{4^{k-1}}{\bf e}_{i}$). 

    \item 
    Next, consider a leaf node $N^{k,l}_{d,j}$ of $\T^{k,l}_B$, where $k<\log_4 M$, $l\in[2^{d(k-1)}]$ and $j\in[2^{d-1}]$. Recall that the two children nodes of $N^{k,l}_{d,j}$ are $N^{k+1,2^d(l-1)+2j-1}_{1,1}$ and $N^{k+1,2^d(l-1)+2j}_{1,1}$, respectively. To ensure that the hypercubes 
    corresponding to the nodes $N^{k+1,2^d(l-1)+2j-1}_{1,1}$ and $N^{k+1,2^d(l-1)+2j}_{1,1}$ are contained inside the intersection of the hypercubes 
    corresponding to the nodes lying in the path from the root $N^{k,l}_{1,1}$ to the leaf $N^{k,l}_{d,j}$, we define $c^{k+1,2^d(l-1)+2j-1}_{1,1}=\frac{c^{k,l}_{1,1}+c^{k,l}_{d,j}}{2}+\frac{M}{2^{2k-1}}{\bf e}_{d}$ and $c^{k+1,2^d(l-1)+2j}_{1,1}=\frac{c^{k,l}_{1,1}+c^{k,l}_{d,j}}{2}-\frac{M}{2^{2k-1}}{\bf e}_{d}$. 
    In other words, the center $c^{k+1,2^d(l-1)+2j-1}_{1,1}$ (respectively, $c^{k+1,2^d(l-1)+2j}_{1,1}$) corresponding to the left (respectively, right) child of $N^{k,l}_{d,j}$ is a translated copy of the point $\frac{c^{k,l}_{1,1}+c^{k,l}_{d,j}}{2}$, differing only in the $d$th coordinate. In particular, for the left (respectively, right) child, this translation is by $\frac{M}{2^{2k-1}}{\bf e}_{d}$ (respectively, $-\frac{M}{2^{2k-1}}{\bf e}_{d}$).
    \end{itemize}

    

    From the construction of hypercubes, it is easy to observe the following lemmas (for proof, we refer to Appendix~\ref{App_obs: homothetic child node relation} and~\ref{App_claim:lb_hyper}).

    \begin{lemma}\label{obs: homothetic child node relation}
        Let $N$, $N'$, and $N''$ be three nodes of the tree $\T$ such that $N'$ and $N''$ are the two children of $N$. Also, let $\sigma'$ and $\sigma''$ be two hypercubes corresponding to the nodes $N'$ and $N''$, respectively. Then, $\sigma'$ and $\sigma''$ are interior disjoint. 
    \end{lemma}


\begin{lemma}\label{claim:lb_hyper}
Let $\T^{k,l}_B$ be a tree for some $k\in[\log_4 M]$ and $l\in[2^{d(k-1)}]$. Let ${\cal J}=(\sigma_u)_{u=1}^{d}$ be a sequence of hypercubes having a width $w=\frac{M}{4^{k-1}}$ corresponding to nodes along a path from the root $N_{1,1}^{k,l}$ to a leaf $N^{k,l}_{d,j}$ of the tree $\T^{k,l}_B$ for some $j\in[2^{d-1}]$. Then, we have the following:
\begin{enumerate}
    \item[(i)] The common intersection region $Q=\cap_{u=1}^{d} \sigma_u$ contains two interior disjoint hypercubes $H_1$ and $H_2$ having width $\frac{w}{2}$ centered at the points $\frac{c^{k,l}_{1,1}+c^{k,l}_{d,j}}{2}+\frac{w}{2}{\bf e}_{d}$ and $\frac{c^{k,l}_{1,1}+c^{k,l}_{d,j}}{2}-\frac{w}{2}{\bf e}_{d}$, respectively.
    
    \item[(ii)] If $k<\log_4 M$, let $\T_B^{k+1,l'}$ be a tree rooted at a child of the leaf node $N^{k,l}_{d,j}$ of the tree $\T_B^{k,l}$, where  $l'\in \{2^d(l-1)+2j-1, 2^d(l-1)+2j\}$. Let ${\Pi}=(\sigma'_u)_{u=1}^{d}$ be a sequence of hypercubes corresponding to a path from the root to a leaf of the tree $\T_B^{k+1,l'}$. Then, the union of the hypercubes $\cup_{u=1}^{d} \sigma'_u$ is contained in $\cap_{u=1}^{d} \sigma_u$.
\end{enumerate}
\end{lemma}

    \noindent  Note that the width of a hypercube corresponding to a leaf node of $\T$ is $4$ (by putting the value of $k=\log_4 M$ in $\frac{M}{4^{k-1}}$). As a result, due to Lemma~\ref{claim:lb_hyper},  the intersection of hypercubes corresponding to a sequence $\I\in \S$ contains a hypercube having a width $2$. Since a hypercube having width $2$ must contain an integer point, we have the following lemma.
    

    \begin{lemma}\label{obs: homothetic opt off}
        Let ${\cal I}=(\sigma_i)_{i=1}^h\in \S$ be a sequence of hypercubes having widths in the range $[1, M]$ corresponding to nodes of a path from the root to a leaf of the tree $\T$, where $h$ is the height of the tree $\T$. Then, their intersection $\cap_{i=1}^h \sigma_i$ must contain an integer point.
    \end{lemma}

    Due to Lemma~\ref{obs: homothetic opt off}, for any sequence $\I\in \S$, the size of the hitting set produced by an optimum offline algorithm over the input sequence~$\I$ is 1. Now, consider any online algorithm ALG. Let $C_i$ be the cost of the algorithm ALG at the $i$th level. Clearly, $\mathbb{E}(C_1)=1$. We will show that $\mathbb{E}(C_i)=\frac{1}{2}$ for $i\in [h]\setminus\{1\}$. While it is possible, in principle, for an algorithm to select several points at each level, it can be seen that this does not help in total cost reduction. W.l.o.g., we can assume that at each level, ALG selects only one point if the currently introduced hypercube is not hit by the solution set maintained by ALG.
    Due to Lemma~\ref{obs: homothetic child node relation}, this point may lie either in the left or in the right children node but not both. Thus, with probability $\frac{1}{2}$, the algorithm ALG chooses either a point or no point at $i$th level. As a result, $\mathbb{E}(C_i)=\frac{1}{2}\times 1 + \frac{1}{2}\times 0= \frac{1}{2}$. Therefore, we have that $\mathbb{E}(\text{ALG})=\sum_{i=1}^h \mathbb{E}(C_i) = 1+\frac{1}{2}(h-1)=\frac{1}{2}(1+\frac{d}{2}\log_2 M)$. Hence,  we have the following theorem.


\begin{theorem}\label{thm: random lower bound homothetic hypercubes}
    The competitive ratio of any (deterministic or randomized) algorithm for hitting similarly sized homothetic hypercubes having widths in the range $[1, M]$ using points in $\IZ^d$ is $\Omega(d\log M)$.
\end{theorem}





\section{Algorithm for Similarly Sized Homothetic Hypercubes in $\IR^d$ Using Points in $\mathbb{Z}^d$} \label{sec: hypercube_upperb}
In this section, we consider the hitting set problem for similarly sized homothetic hypercubes in $\IR^d$ having {widths} in the range $[1, M]$ using points in $\mathbb{Z}^d$. Note that an axis-aligned hypercube's width is half of its side length.
For this, we propose an $O(d^2\log M)$ competitive randomized algorithm $\LIR$. Throughout this section, if explicitly not mentioned, all distances are $L_{\infty}$ distances, and all hypercubes are axis-aligned. We use $\Q_s(\sigma)$ to denote the set $\sigma \cap (s\mathbb{Z})^d$.

For different values of $s\in \Z$, our algorithm $\LIR$ evokes a randomized iterative reweighting algorithm  $\RIR_s$ as a subroutine. The algorithm $\RIR_s$ considers a sequence of hypercubes of width $s$ as input and uses points of $(s\Z)^d$ to hit them. This algorithm is a generalization of an algorithm in~\cite{DeS24b}.  Here, first, we present a description of the algorithm $\RIR_s$ below.

\vspace{1mm}\noindent\textbf{Algorithm $\RIR_s$.} Let ${\cal I}^{(s)}$ be the set of hypercubes having width $s$, presented to the algorithm $\RIR_s$. Let $\A^{(s)}$ be the set of points chosen by our algorithm such that each hypercube in $\I^{(s)}$ contains at least one point from $\A^{(s)}$. The algorithm maintains two disjoint sets $\A_1^{(s)}$ and $\A_2^{(s)}$ such that $\A^{(s)}= \A_1^{(s)} \cup \A_2^{(s)}$. The algorithm also maintains another set $\B^{(s)}$ of points for bookkeeping purpose; initially, each of the sets $\I^{(s)},\A^{(s)}\text{ and }\B^{(s)}$ are empty. A weight function  $w$ over all integer points from ${(s\mathbb{Z})}^d$ is also maintained by the algorithm; initially, $w(p)=3^{-(d+1)}$, for all points $p \in (s\mathbb{Z})^d$. One iteration of the algorithm is as follows.
Let $\sigma$ be a new hypercube of width $s$; update $\I^{(s)} = \I^{(s)} \cup\{\sigma\}$. 
\begin{enumerate}
    \item If the hypercube $\sigma$ contains a point from $\A^{(s)}$, then do nothing.
    \item Else if the hypercube $\sigma$ contains a point from $\B^{(s)}$, then let $p\in \B^{(s)}\cap\Q_s(\sigma)$ be an arbitrary point, and update $\A_1^{(s)} = \A_1^{(s)} \cup\{p\}$.
    \item Else if $\sum_{p\in \Q_s(\sigma)} w(p)\geq 1$, then let $p$ be an arbitrary point in $\Q_s(\sigma)$, and update $\A_2^{(s)} = \A_2^{(s)} \cup\{p\}$.
    \item Else, the weights give a probability distribution on $\Q_s(\sigma)$. Successively choose points from $\Q_s(\sigma)$ at random with this distribution in $\lceil\frac{5d}{2}\rceil$ independent trials and add them to $\B^{(s)}$. Let $p\in \B^{(s)}\cap\Q_s(\sigma)$ be an arbitrary point, and update $\A_1^{(s)} = \A_1^{(s)} \cup\{p\}$. Triple the weight of every point in $\Q_s(\sigma)$.
\end{enumerate}

\begin{lemma}\label{thm:rir}\cite[Theorem~8]{DeS24b}
{For a given $s\in \Z$, the randomized algorithm $\RIR_s$ achieves a competitive ratio of~$O(d^2)$ for hitting hypercubes having width $s$ using points in $(s\mathbb{Z})^d$, where~$d\geq3$.} 
\end{lemma}

To prove the above, the following fact is used. For each hypercube $\sigma$, the value of $|\Q_s(\sigma)|$ is at least $2^d$ and at most $3^d$. This leads us to the following two critical conditions.
\vspace{1mm}
\begin{enumerate}
    \item \textbf{Layer Formation.}  The set of homothetic hypercubes needs to be partitioned into layers such that each hypercube within a particular layer contains at least $2^d$ points and at most $3^d$ points from $(s\mathbb{Z})^d$, where the value of $s$ will depend on the layer of the object.
     \item \textbf{Core Construction.} For each object $\sigma$ in a given layer, a core hypercube $cr(\sigma)$ needs to be constructed with a fixed  width $s$ such that the core hypercube  $cr(\sigma)$ contains the same set of points from $(s\mathbb{Z})^d$ as the hypercube $\sigma$, where the value of $s$ will depend on the layer of the object.
\end{enumerate} 

Let's address the first condition, i.e., the layer formation. 
We partition the set of all homothetic hypercubes into $2(\log M+1)$ layers such that for each $k\in[2\log M+1]\cup\{0\}$, when $k$ is even (respectively, odd), the layer $L_k$ contains the hypercubes having widths in the range $[2.2^{\frac{k}{2}-1},3.2^{\frac{k}{2}-1})$ (respectively $[3.2^{\frac{k-1}{2}-1},4.2^{\frac{k-1}{2}-1})$). 
Let $s(k)=2.2^{\frac{k}{2}-1}$ if $k$ is even, and $s(k)=3.2^{\frac{k-1}{2}-1}$ if $k$ is odd. 
We hit hypercubes in $L_k$ using points from $\mathbb{Z}_k=(s(k)\mathbb{Z})^d$.

Due to Observation~\ref{obs:points}, a hypercube $\sigma\in L_k$ contains at least $2^d$ and at most $3^d$ points from $\mathbb{Z}_k$. Since the objects in a layer are not translated copies of each other, we still cannot use the algorithm $\RIR_{s(k)}$ for hitting objects in layer $L_k$. 
Thus, we introduce the concept of core $cr(\sigma)$ of a hypercube $\sigma$ as follows. 
Let $\sigma\in L_k$. A core $cr(\sigma)$ of $\sigma$ is a hypercube having a width of $s(k)$ such that  $cr(\sigma)$ contains the same set of points from $\mathbb{Z}_k$ as the hypercube $\sigma$, i.e., $\Q_{s(k)}(\sigma)=\Q_{s(k)}(cr(\sigma))$. Notice that to hit an object $\sigma\in L_k$, it is sufficient to hit the core $cr(\sigma)$. 
This tackles the second condition.





\begin{figure}[htbp]
    \centering
    \includegraphics[page=3, width=80mm]{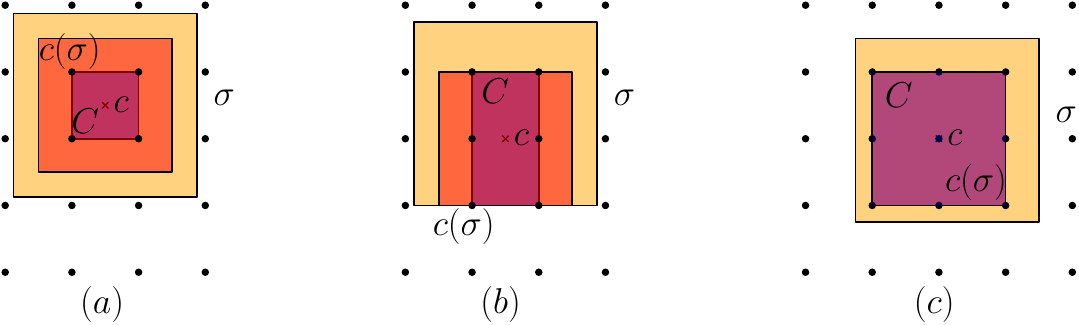}
    \caption{Illustration in plane to determine the core $cr(\sigma)$ of a hypercube $\sigma\in L_k$. Here, the black colored points denote the integer points in $\mathbb{Z}_k$, where $k\in\mathbb{Z}^{+}\cup\{0\}$. Here, the orange-colored square denotes $\sigma$. The red-colored square denotes the core $cr(\sigma)$, centered at $c$, of $\sigma$. The purple-colored rectangle denotes the convex hull $C$, centered at $c$, of integer points within $\Q_{s(k)}(\sigma)$.  When $\sigma$ contains (a) 4 (b) 6 (c) 9 (both $C$ and $cr(\sigma)$ are the same) points from $\mathbb{Z}_k$.}
    \label{fig:core}
\end{figure}
Consider a hypercube $\sigma$ belonging to a layer $L_k$.
Let $z_1,z_2,z_3,\ldots,z_l$ be $l$ integer points from~$\IZ_k$ satisfying the following condition. There exists an index $k'\in[d]$ such that $|z_i(x_{k'})-z_{i+1}(x_{k'})|=s(k)$, where $i\in[l-1]$,  and for all other indices $k''\in[d]\setminus\{k'\}$, we have  $z_i(x_{k''})=z_{i+1}(x_{k''})$, where $i\in[l-1]$. We term these points as \emph{successively consecutive points} from {$\IZ_k$}. Observe that $d_{\infty}(z_1,z_4)=3s(k)$. Since the width of $\sigma$ is strictly less than $\frac{3}{2}s(k)$, the maximum distance between any two points in $\sigma$ is strictly less than $3s(k)$. Thus, the hypercube $\sigma$ cannot contain four successively consecutive points from $\IZ_k$. Thus, we have the following.

\begin{claim}\label{clm_pair}
The distance between any pair of points of $\Q_{s(k)}(\sigma)$ is at most $2s(k)$. 
\end{claim}

Let $C$ be the convex hull of points in $\Q_{s(k)}(\sigma)$. Due to Claim~\ref{clm_pair}, it is easy to observe that the convex hull $C$ forms a hyperrectangle whose each side length is either $s(k)$ or $2s(k)$. Let $c$ be the center of $C$. Let $H$ be a closed hypercube centered at $c$ having side length $2s(k)$. Now, we claim the following (refer to Appendix~\ref{App_clm_core} for a proof).

\begin{claim}\label{clm_core}
   $\Q_{s(k)}(\sigma)=\Q_{s(k)}(H)$.
\end{claim}

This hypercube $H$ represents the core $cr(\sigma)$ of $\sigma$. For an illustration in $\IR^2$, see Figure~\ref{fig:core}. 
Now, we describe our algorithm $\LIR$ for hitting homothetic hypercubes having widths in the range $[1, M]$.

\vspace{1mm}\noindent{\bf{Algorithm $\LIR$}}.\label{algo_LIR}
The algorithm maintains a set $\A^{(s(k))}$ for each layer $L_k$ that has been part of the input so far.
The algorithm maintains the set $\A=\cup_{k}\A^{(s(k))}$ as the solution set.
Upon the arrival of a hypercube $\sigma$ having a width $w$, if it is already hit by $\A$, then we ignore it. 
Otherwise, our algorithm, first, determines the layer number $k=\lfloor\log_{\frac{3}{2}} w\rfloor$. 
Next, consider the core $cr(\sigma)$ as an input to $RIR_{s(k)}$ and
update $\A^{(s(k))}$ appropriately. 



\begin{theorem}\label{thm: hyp_ub}
    For hitting similarly sized homothetic hypercubes having widths in the range $[1, M]$, for some $M>1$, using points in $\mathbb{Z}^d$, the randomized algorithm $\LIR$ achieves a competitive ratio of~$O(d^2\log M)$, where $d\geq3$.
\end{theorem}
\begin{proof}
   Let $\cal I$ be a collection of homothetic hypercubes presented to the online algorithm.
    For each $k\in[2\lfloor \log M\rfloor+1]\cup\{0\}$, let ${\cal I}_k$ be the collection of all hypercubes in ${\cal I}$ belonging to the layer $L_k$. Let ${\cal C}_k$ be the collection of core $cr(\sigma)$ of $\sigma\in{\cal I}_k$ having widths $s(k)$. Let $\OO$ be an optimum offline hitting set for $\cal I$. Let $p$ be a point in  $\OO$.
    For hitting core objects from layer ${\cal C}_k$, due to Lemma~\ref{thm:rir}, corresponding to the point $p$, the expected number of points placed from $\IZ_k$ by the algorithm $\RIR_{s(k)}$ is~$O(d^2)$. Notice that there are total $2\lfloor \log M\rfloor+2$ layers of core objects.
    As a result, for hitting hypercubes in ${\cal I}$ corresponding to the point $p$, the expected number of points placed from $\mathbb{Z}^d$ by the algorithm $\LIR$ is~$O(d^2\log M)$.
    Hence, we have $\mathbb{E}[|\A|]=O(d^2\log M)|\OO|$. Thus, the theorem follows.
\end{proof}

\section{Algorithm for Similarly Sized Fat Objects in $\IR^d$ Using Points in $\mathbb{Z}^d$}\label{sec: fat}
In this section, we study the hitting set problem for a range space $\Sigma=(\P, \cal S)$, where the point set $\P=\mathbb{Z}^d$ and the set ${\cal S}$ is a family of similarly sized $\alpha$-fat objects in $\IR^d$ having widths in the range $[1, M]$, where $M> 1$. For this, we propose a deterministic online algorithm $\ANC$ and prove that the algorithm achieves a competitive ratio of at most~${\lfloor\frac{2}{\alpha}+2\rfloor^d} \left(\lfloor\log_{2}M\rfloor+1\right)$.
Throughout the section, if explicitly not mentioned, all distances are $L_{\infty}$ distances, and all hypercubes are axis-aligned.
For simplicity, we use the term object to represent $\alpha$-fat object.

\vspace{1mm}
\noindent\textbf{Algorithm $\ANC$.}
For each $i\in\mathbb{Z}^{+}\cup\{0\}$, let $L_i$ be the $i$th layer containing similarly sized objects having widths in the range $[2^i,2^{i+1})$.
Upon the arrival of an object $\sigma$ centered at $c$ and having width $w$, we do the following
\vspace{1mm}\begin{enumerate}
    \item If $\sigma$ is stabbed by the existing hitting set, do nothing.
    \item Else, first, determine the layer {$L_i$} in which $\sigma$ belongs, where $i=\lfloor\log w\rfloor$. To hit $\sigma$, our online algorithm adds a point $r$ from  $ (2^{i+1}\mathbb{Z})^d$ as follows.
\end{enumerate} 
Note that, for each {$j\in[d]$}, the $j$th coordinate of the point $c$ can  be uniquely written as $c(x_j)=z_j+f_j$, where $z_j\in 2^{i+1}\mathbb{Z}$ and $f_j\in\left[0,2^{i+1}\right)$. 
Based on the coordinates of the center $c$ of $\sigma$, we define the point $r$.
For each $j\in[d]$, we set the $j$th coordinate of $r$ as follows. If $f_j\in\left[0,2^{i}\right)$, then $r(x_j)=z_j$, and if $f_j\in\left[2^{i},2^{i+1}\right)$, then $r(x_j)=z_j+2^{i+1}$.
Note that each coordinate of $r$ is an integral multiple of $2^{i+1}$. Thus, $r\in(2^{i+1}\mathbb{Z})^d$.
As per the construction of {the point} $r$, we have $|r(x_j)- c(x_j)|\leq 2^{i}$ for each $j\in[d]$. As a result, we have the following.

\begin{claim}\label{clm:drc}
    $d_{\infty}(r,c)=\max_{j\in[d]}|r(x_j)-c(x_j)|\leq 2^{i}$.
\end{claim}

One can prove that the  point $r$ is a closest point of $(2^{i+1}\mathbb{Z})^d$ from $c$ (for proof see Appendix~\ref{App_clm: r closest point}).
%
%
%
\noindent
The following lemma ensures the correctness of our algorithm. One can prove this using Claim~\ref{clm:drc} (refer to Appendix~\ref{App_lem_correct}  for the proof).
\begin{lemma}\label{lem_correct}
   The object $\sigma$ contains the  point $r$ from $(2^{i+1}\mathbb{Z})^d$.
\end{lemma}


    

\begin{theorem} \label{thm:fat_ub}
    The algorithm $\ANC$ achieves a competitive ratio of at most ${\lfloor\frac{2}{\alpha}+2\rfloor^d}$ $\left(\lfloor\log_{2}M\rfloor+1\right)$ for hitting similarly sized $\alpha$-fat objects in $\IR^d$ having  widths in the range $[1,M]$ using points in $\mathbb{Z}^d$.
\end{theorem}


\vspace{-5mm}
\begin{proof}
    Let  $\I$ be a collection of similarly sized objects presented to the online algorithm.
    For each $i\in\{0,1,\ldots,\lfloor \log M\rfloor\}$, let $\I_i$ be the collection of all objects in $\I$ having widths in the range $\left[2^i, 2^{i+1}\right)$. Let $\A$ and $\OO$ be hitting sets returned by the algorithm $\ANC$ and an optimum offline algorithm, respectively, for the input sequence~$\I$. Let $p$ be a point in $\OO$. Let $\I_p\subseteq \I$ be the set of objects containing the point $p$. For each $i\in\{0,1,\ldots,\lfloor \log M\rfloor\}$, let $\I_{p,i}\subseteq \I_p$ be the collection of all objects having widths in the range $\left[2^i, 2^{i+1}\right)$. Let $\A_{p}\subseteq \A$ be the set of hitting points placed by $\ANC$ to hit all objects in $\I_p$. For each $i\in\{0,1,\ldots,\lfloor \log M\rfloor\}$, let $\A_{p,i}\subseteq \A_p$ be the set of hitting points explicitly placed by our algorithm $\ANC$ to hit objects in $\I_{p,i}$.
    It is easy to see that $\A=\cup_{p\in \OO}\A_p=\cup_{p\in \OO}\left(\cup_{i=0}^{\lfloor\log M\rfloor}\A_{p,i}\right)$.
    Therefore, the competitive ratio of our algorithm is upper bounded by $|\A_p|$.

   Let $c$ be the center of an object $\sigma\in \I_{p,i}$. To hit $\sigma$, our algorithm adds a point $r\in(2^{i+1}\mathbb{Z})^d$ such that $d(r,c)\leq 2^i$ (due to Claim~\ref{clm:drc}). Since $c$ is the center of $\sigma\in \I_{p,i}$ having a width strictly less than $2^{i+1}$ and $p\in\sigma$, we have $d(c,p)<\frac{2^{i+1}}{\alpha}$.
    Now, using triangle inequality, we have $d(r,p)\leq d(r,c)+d(c,p)$. Consequently, we have $d(r,p)\leq \frac{2^{i+1}}{\alpha}+ 2^{i}$.
     Hence, the interior of a hypercube of side length $\frac{2^{i+2}}{\alpha}+2^{i+1}$ centered at $p$ contains all points in $\A_{p,i}$. 
Due to Observation~\ref{obs:points}, the interior of any hypercube of side length $2^{i+1}(\frac{2}{\alpha}+1)$ contains at most $\lfloor\frac{2}{\alpha}+2\rfloor^d$ points from $(2^{i+1}\mathbb{Z})^d$.
Thus, we have $|\A_{p,i}|\leq \lfloor\frac{2}{\alpha}+2\rfloor^d$. Recall that $|\A_p|=\Sigma_{i=0}^{\lfloor\log_2 M\rfloor}|\A_{p,i}|$. Thus, we have $|\A_p|\leq \lfloor\frac{2}{\alpha}+2\rfloor^d(\lfloor\log_2 M\rfloor+1)$. Hence, the theorem follows.
\end{proof}
    




For balls in $\IR^d$, by putting the value of  $\alpha=\frac{1}{\sqrt{d}}$  in Theorem~\ref{thm:fat_ub}, we have the following.

\begin{corollary}\label{cor:ball}
    For hitting balls in $\IR^d$ having  radii in the range  $[1,M]$ using points in $\mathbb{Z}^d$, the algorithm $\ANC$ achieves a competitive ratio of at most ${\lfloor2\sqrt{d}+2\rfloor^d} \left(\lfloor\log_{2}M\rfloor+1\right)$.
\end{corollary}

\section{Algorithm for Similarly Sized Homothetic Regular $k$-gons Using $n$  Points in $\IR^2$}\label{sec: homothetic regular k-gons}
In this section, we present an algorithm for the online hitting set problem for a range space $\Sigma=(\P,\cal S)$, where the set $\P\subset\IR^2$ is a collection of $n$ points and the set $\cal S$ is a family of similarly sized homothetic regular $k$-gons having diameters in the range $[1, M]$, where $M> 1$. Here, only the point set $\P$ is known in advance, and the objects in $\S$ will be introduced one by one. For this problem, we propose an $O(\log n\log M)$ competitive randomized algorithm $\HHR$. Note that our algorithm does not need to know the value of $M$ in advance. Algorithm $\HHR$ uses a deterministic algorithm $\ES$ (similar to the one in~\cite[Section~7]{EvenS14}) as a subroutine.

Before describing the algorithms, we first define some notations. Throughout this section, by a regular $k$-gon, we mean a regular $k$-gon  with  $k\geq 4$. The \emph{center} and \emph{diameter} of a regular $k$-gon $\sigma$ are the center and diameter of the smallest disk enclosing $\sigma$, respectively.
We use the symbol $\sigma(o)$ to represent an object $\sigma$ centered at a point $o$. 
{We use $-\sigma$ to denote a reflected copy of $\sigma$ through the origin (for an example, see Figure~\ref{fig:ref} in Appendix~\ref{sec: Structural Property of Regular k-gons}).}



 A \emph{vertex ranking} of a graph $G=(V,E)$ is a coloring $c: V \rightarrow [k]$ such that for any two distinct vertices $x$ and $y$, and for every simple path $P$ from $x$ to $y$, if $c(x)=c(y)$, then there exists an internal vertex $z$ in $P$ such that $c(x)<c(z)$. 
The \emph{vertex ranking number} of $G$ is the smallest integer $k$ such that $G$ admits a vertex ranking using only $k$ colors.

\subsection{Algorithm $\mathbf{\ES}$}\label{sec: algo-es}
This $\ES$ is an algorithm for the hitting set problem for translated copies of a regular $k$-gon using a fixed set of $n$ points in $\IR^2$. 

\vspace{1mm}\noindent\textbf{Preprocessing Step.}
As part of preprocessing, we do the following. (i) Consider a partition $P_{\sigma}$ of the plane. (ii) For every tile $s$ in $P_{\sigma}$ such that $s\cap \P\neq \emptyset$, compute four sets of extreme points $V_{s,{\tau}}$  for $\tau\in [4]$.
(iii) For each $V_{s,{\tau}}$, construct a path graph $G_{s,\tau}$ over the set $V_{s,{\tau}}$ and then, compute a vertex ranking $c_{s,\tau}$ of the graph $G_{s,\tau}$ that uses $\lfloor\log_2(2|V_{s,\tau}|)\rfloor$ colors~\cite[Proposition~5]{EvenS14}.

Now, we will describe the first two preprocessing steps in detail. The third step, i.e., computing a vertex ranking,
is the same as described in~\cite{EvenS14}.

\vspace{1mm}\noindent\hspace{2mm} $\bullet$ \underline{Partitioning the Plane}:
Let $\sigma$ be a regular $k$-gon. Depending on the shape and size of the object $\sigma$, we define a parameter $\ell_{\sigma}$. Assuming the diameter of $\sigma$ to be 2 units, we set the value of $\ell_{\sigma}$ as follows. If $\sigma$ is a square (respectively, regular $k$-gon with $k\geq 7$ and   regular $k$-gon with $k\in\{5,6\}$),  the value of $\ell_{\sigma}$ is $\frac{1}{2}$ (respectively, $\frac{1}{2}$ and $\frac{1}{4}$).
Similar to~\cite{ChenKS09}, we partition the plane using square tiles each of side length $\ell_{\sigma}$ such that each point $p\in\P$ lies in the interior of a tile. It is easy to achieve such a partition, $P_{\sigma}$, since the set $\P$ is finite.
Let $m_{\sigma}$ be the maximum number of tiles of the partition $P_{\sigma}$ that any translated copy of $\sigma$ can intersect. 
We refer to Appendix~\ref{sec: appendix algo-es} for proof of the following observation.

\begin{observation}\label{obs:contain}
    The value of $m_{\sigma}\leq 25$ $($respectively, $23$ and $63)$, when  $\sigma$ is a square $($respectively, regular $k$-gon with $k\geq 7$ and  regular $k$-gon with $k\in\{5,6\})$.
\end{observation}

\noindent\hspace{2mm} $\bullet$ \underline{Extreme Points}: 
Depending on the shape and size of the object $\sigma$, we define a parameter~$t_{\sigma}$. Assuming the diameter of $\sigma$ as 2 units, we set the value of $t_{\sigma}$ as follows. If $\sigma$ is a square (respectively, regular $k$-gon with $k\geq 7$ and regular $k$-gon with $k\in\{5,6\}$), the value of $t_{\sigma}$ is~$\frac{5}{2}$ $\left(\text{respectively, } \frac{5}{2} \text{ and } \frac{9}{4}\right)$.
For each square tile $s$ from the partition $P_{\sigma}$, we define a \emph{super-square} $S$ that is a square concentric with the tile $s$ whose side length is $t_{\sigma}$. Let ${\cal D}_s$ be the set of all translated copies of $\sigma$ with non-empty intersection with the tile $s$.
It is easily observed that $S$  contains centers of all objects in ${\cal D}_s$. 
We divide the super-square $S$ into four quadrants, each square of side length $t_{\sigma}/2$.  For $\tau\in [4]$, let  $S_{\tau}$ denote a quadrant of $S$, and let $o_{\tau}$ denote its center (see Figure~\ref{fig:cone}).
For proof of the following observation, see Appendix~\ref{sec: appendix algo-es}.

\begin{observation}\label{obs:inter}
If $\sigma' \in{\cal D}_s$, then  $\sigma'\cap\{o_1,o_2,o_3,o_4\}\neq\emptyset$.
\end{observation}

For $\sigma'\in{\cal D}_s$, let $\tau_{s,{\sigma'}}=\min\{\tau\ |\ o_{\tau}\in \sigma'\}$. Due to Observation~\ref{obs:inter}, the value of $\tau_{s,{\sigma'}}$ is well-defined. Let ${\cal D}_{s,{\tau}}$ denote the set $\{ \sigma'\in{\cal D}_s\ |\ \tau_{s,{\sigma'}}=\tau\text{ and } \sigma'\cap s\cap\P\neq\emptyset\}$. 
For every tile $s$ from the partition $P_{\sigma}$ and every $\tau\in[4]$, we define a set $V_{s,\tau}$ of extreme points as follows:
$V_{s,{\tau}}=\{p\in\P\ |\ \exists\  \sigma'\in {\cal D}_{s,{\tau}}\text{ such that } p\in\partial\sigma' \text{ and } \sigma'\cap s\cap\P=\{p\}\}$ (see Figure~\ref{fig:extreme_point}).
Note that if $ \sigma'\in {\cal D}_{s,{\tau}}$ and $\sigma'\cap s\cap\P\neq\emptyset$, then $\sigma'\cap V_{s,{\tau}}\neq\emptyset$.


\vspace{1mm}
\noindent\emph{Identifying Extreme Points.}
Here, we use \emph{$d_{C}(.,.)$ function} (also known as \emph{convex distance function}), induced by a convex object $C$, and its following property~\cite{DeJKS24, IckingKLM95}.

\begin{property}\label{prop:conv}~\cite{DeJKS24}
    Let $C$ be a convex object centered at a point $c$, and let $x,y$ be two points in $\IR^d$. Then, we have the following properties:
    \begin{enumerate}
        \item[(i)] $x\in C$ if and only if $d_C(c,x)\leq 1$. In particular, $x\in \partial C$ if and only if $d_C(c,x)=1$.
        \item[(ii)] $d_C(x,y)=d_{-C}(y,x)$, where $-C$ is a reflected copy of $C$ through the origin.
    \end{enumerate}
\end{property}

Let $p\in\P\cap s$ be a point, and $o_{\tau}$ be the center of a quadrant of the super-square of the tile $s$. 
Let $\bm{\sigma'}(o_{\tau})$ and ${\bm\sigma'}({p})$ be two copies of $-\sigma$, centered at $o_{\tau}$ and  $p$, respectively. {Now, we define the Voronoi diagram for the points in $\P\cap s$ with respect to the $d_{\sigma}$ function, induced by the convex object $\sigma$}.
We define \emph{Voronoi cell} of $p\in \P\cap s$ as $\Cell({p})=\{x\in \IR^2\ |\  {d}_{\sigma}(x,p)\leq {d}_{\sigma}(x,q) \text{~for all } q(\neq p)\in \P\cap s\}$. 
The \emph{Voronoi diagram} is a partition on the plane imposed by the collection of cells of all the points in $\P\cap s$.
Note that one can compute this Voronoi diagram in $O(|\P\cap s|\log(|\P\cap s|))$ time~\cite{ChewD85}.
A point $x$ is in $\Cell({p})$  implies that if we place a homothetic (i.e., translated and scaled) copy $\sigma'$ of $\sigma$ centered at $x$ with a scaling factor $d_{\sigma}(x,p)$, then $\sigma'$ touches  $p$, and the interior of $\sigma'$ does not contain any point of $\P\cap s$. The following lemma offers a method to decide whether $p$ is an extreme point (without knowing the input set $\S$).
\begin{lemma}\label{clm:ext}
A point $p\in \P\cap s$ belongs to $V_{s,\tau}$ if and only if {$\bm{\sigma'}(o_{\tau})\cap \partial {\bm\sigma'}(p)\cap \interior{(\Cell(p))}$} is non-empty.
\end{lemma}
\begin{proof}
For the forward part, let us assume that $p\in V_{s,\tau}$. Due to the definition of  extreme points, {there exists an object $\sigma'\in D_{s,\tau}$} such that 
$p\in\partial\sigma'$ and $ \sigma'\cap s\cap\P=\{p\}$. Let $c$ be the center of $\sigma'$. To show $\bm{\sigma'}(o_{\tau})\cap \partial {\bm\sigma'}(p)\cap \interior{(\Cell(p))}$ is non-empty, we will show that $c\in\bm{\sigma'}(o_{\tau})\cap \partial {\bm\sigma'}(p)\cap \interior{(\Cell(p))}$. 
Since $\sigma'$ does not contain any point other than $p$ from $\P\cap s$, the distance $d_{\sigma}(c,p)$ is strictly less than $d_{\sigma}(c,q)$, where $q(\neq p)$ is any point in $\P\cap s$. As a result, we have $c\in \interior(\Cell(p))$.
Since $o_{\tau}\in \sigma'$ (respectively, $p\in\partial\sigma'$ ), we have $d_{\sigma}(c,o_{\tau})\leq1$ (respectively, $d_{\sigma}(c,p)=1$).
Due to Property~\ref{prop:conv}, we have $d_{-\sigma}(o_{\tau},c)=d_{\sigma}(c,o_{\tau})$ (respectively, $d_{-\sigma}(p,c)=d_{\sigma}(c,p)$). Hence, we have $d_{-\sigma}(o_{\tau},c)\leq 1$ (respectively, $d_{-\sigma}(p,c)=1$). Thus,  $c\in\bm{\sigma'}(o_{\tau})$ (respectively, $c\in\partial {\bm\sigma'}(p)$). As a result, we have $c\in\bm{\sigma'}(o_{\tau})\cap\partial {\bm\sigma'}(p)\cap\interior(\Cell(p))$. 

For the converse part, let us assume that $\bm{\sigma'}(o_{\tau})\cap \partial {\bm\sigma'}(p)\cap \interior{(\Cell(p))}$ is non-empty. Let $c\in\bm{\sigma'}(o_{\tau})\cap \partial {\bm\sigma'}(p)\cap \interior{(\Cell(p))}$ be a point.
Since $c\in \bm{\sigma'}(o_{\tau})\cap \partial {\bm\sigma'}(p)$, we have $d_{-\sigma}(o_{\tau},c)\leq1$ and $d_{-\sigma}(p,c)=1$.
Due to Property~\ref{prop:conv}, we have $d_{-\sigma}(o_{\tau},c)=d_{\sigma}(c,o_{\tau})$ (respectively, $d_{-\sigma}(p,c)=d_{\sigma}(c,p)$). Thus,  $d_{\sigma}(c,o_{\tau})=d_{-\sigma}(o_{\tau},c)\leq 1$ (respectively, $d_{\sigma}(c,p)=d_{-\sigma}(p,c)=1$). Hence, $\sigma(c)$, a translated copy of $\sigma$ centered at $c$, contains both $o_{\tau}$ and $p$ ($\in \partial \sigma(c)$). Now, we only need to show that $\sigma(c)$ contains only $p$ from $\P\cap s$. Since $c\in \partial {\bm\sigma'}(p)\cap \interior{(\Cell(p))}$, the distance $d_{\sigma}(c,p)\ (= 1)$ is strictly less than $d_{\sigma}(c,q)$, where $q(\neq p)$ is any other point in $\P\cap s$. As a result, we have $\sigma(c)\cap\P\cap s=\{p\}$. Hence, $p$ is an extreme point.
 \end{proof}

\noindent\emph{Computing Extreme Points.} Here, we discuss how to construct the set $V_{s,\tau}$ of extreme points.
{First, draw a Voronoi diagram with respect to the $d_{\sigma}$ function, induced by the convex object $\sigma$,  for the points in $\P\cap s$~\cite{ChewD85}}. Let $\bm{\sigma'}(o_{\tau})$ be a translated copy of $-\sigma$ centered at $o_{\tau}$.
For each point $p\in\P\cap s$,  consider a translated copy ${\bm\sigma'}({p})$ of $-\sigma$, centered at  $p$. Now due to Lemma~\ref{clm:ext}, if {$\bm{\sigma'}(o_{\tau})\cap \partial {\bm\sigma'}(p)\cap \interior{(Cell(p))}$} is non-empty, then $p$ is an element of $V_{s,\tau}$.


\vspace{1mm}
Now, for the sake of completeness, we describe the algorithm $\ES$  which is essentially the same as~\cite{EvenS14}.

\vspace{1mm}\noindent\textbf{Description of Algorithm $\ES$.}
The algorithm maintains a hitting set $\H_{i-1}$ for all the $i-1$ objects $\{\sigma_1,\sigma_2,\ldots, \sigma_{i-1}\}$ that have been the part of input so far. Upon the arrival of a new object $\sigma_i$, if it is already stabbed  by $\H_{i-1}$, then simply update $\H_i\leftarrow \H_{i-1}$. Otherwise, from each tile $s$ such that $\sigma_i\cap\P\cap s\neq\emptyset$, first determine the $\tau$ such that $\tau=\tau_{s,{\sigma_i}}$ and then select the point $v_{s,\sigma_i}$ such that $c_{s,\tau} (v_{s,\sigma_i})= \max\{c_{s,\tau}(v)\ |\ v \in \sigma_i\cap\P \cap V_{s,{\tau}}\}$. These points are added to $\H_{i-1}$ to obtain $\H_i$.

\vspace{1mm}
In the following lemma, it is demonstrated that upon the arrival of two unstabbed translated copies of $\sigma$ that have a common point $p\in s$ and are of the same type $\tau$, then they are stabbed by extreme points in $V_{s,{\tau}}$ of different colors (for proof see Appendix~\ref{sec: appendix algo-es}). This lemma plays a vital role in analysing the competitive ratio of our algorithms.

\begin{lemma}\label{lem:distinct}
    For any $\tau\in[4]$, if  $\sigma_i,\sigma_j\in{\cal D}_{s,{\tau}}$ are unstabbed upon their arrivals, and there exists a point $p\in \sigma_i\cap \sigma_j\cap\P \cap s$, then $c_{s,\tau}(v_{s,\sigma_i})\neq c_{s,\tau}(v_{s,\sigma_j})$.
\end{lemma}



\begin{figure}[htbp]
    \centering
    \begin{subfigure}[b]{0.30\textwidth}
        \centering
        \includegraphics[page=3, width=30 mm]{Hiiting_Homothetic_kgons.pdf}
        \caption{}
        \label{fig:cone}
    \end{subfigure}
    \begin{subfigure}[b]{0.30\textwidth}
        \centering
        \includegraphics[page=1, width=30 mm]{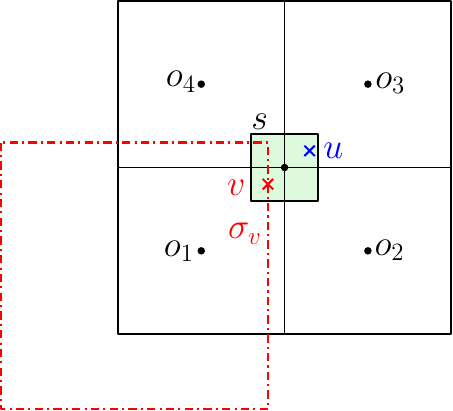}
        \caption{}
        \label{fig:extreme_point}
    \end{subfigure}
    \begin{subfigure}[b]{0.30\textwidth}
        \centering
        \includegraphics[page=2, width=30 mm]{Hiiting_Homothetic_kgons.pdf}
        \caption{ }
        \label{fig:bounded_scaled}
    \end{subfigure}
    \caption{(a) Super-square $S$ of a tile  $s$ of the partition $P_{\sigma}$, where $\sigma$ is a unit square. (b) Here, the point $v$ is an extreme point i.e. $v\in V_{s,1}$ due to the existence of a unit square $\sigma_v\in {\cal D}_{s,1}$. However, the point $u$ does not belong to $V_{s,1}$ because any unit square containing the points $u$ and $o_1$ must contain the point $v$. (c) Illustration of the shrunk set for $\sigma$, when $\sigma$ is not an integral power of $2$.}
    \label{fig:notation222}
\end{figure}



\noindent\emph{Angle Property.} To prove the analog result of Lemma~\ref{lem:distinct}, Even and Smorodinsky~\cite{EvenS14} used the following structural property of disks: {\it  Let $\sigma$ and $\sigma'$ be two distinct translated copies of a disk having a non-empty intersection region. Let $p$ be any point in $\sigma\cap \sigma'$. Let $x$ and $y$ be any two points in $\partial \sigma\cap \partial \sigma'$. Then, the angle $\angle xpy$ of the triangle $\triangle xpy$ is at least~$\frac{\pi}{2}$}.
However, to the best of our knowledge, a similar property was not known for translated copies of a regular $k$-gon. Here, we state a similar property for them. Before saying this property, we introduce some terminologies. We use the term \emph{connected component} to denote a maximal continuous set that cannot be covered by the union of two disjoint open sets (see Figure~\ref{fig:case_6} in Appendix~\ref{sec: Structural Property of Regular k-gons}).
Two objects $\sigma$ and $\sigma'$ are said to \emph{intersect} each other if $\sigma\cap\sigma'$ has some non-empty interior.
Now, we state the angle property for regular $k$-gon (for proof, we refer to Appendix~\ref{sec: Structural Property of Regular k-gons}).

\begin{lemma}[Angle Property]\label{lemma:regular}
    Let $\sigma$ and $\sigma'$ be two distinct translated copies of a regular $k$-gon such that their interior has a non-empty intersection. Let $p$ be any point in $\sigma\cap \sigma'$. Let $x$ and $y$ be any two points in $\partial \sigma\cap \partial \sigma'$ such that $x$ and $y$ belong to two different connected components of $\partial \sigma\cap \partial \sigma'$. Then, the angle $\angle xpy$ of the triangle $\triangle xpy$ is at least~$\frac{\pi}{4}$.
\end{lemma}

Using the Lemma~\ref{lem:distinct}, in a similar way as in~\cite{EvenS14}, one can show the following:

\begin{theorem}\label{thm:translates}
The algorithm $\ES$ achieves a competitive ratio of at most~$4m_{\sigma}\lfloor\log_2 2n\rfloor$ for the hitting set problem of translated copies of a regular $k$-gon $\sigma$, for $(k\geq 4)$, using $n$ points in $\IR^2$, where $m_{\sigma}$ is a constant that depends on $k$.
\end{theorem}

\subsection{Algorithm $\HHR$}
In this section, we describe the algorithm $\HHR$.
 Before describing the algorithm, we introduce some terminologies. For each $j\in\mathbb{Z}^{+}\cup\{0\}$, let $L_j$ be the $j$th layer containing all objects with diameters in the range $[2^j,2^{j+1})$. 
Now, we present the preprocessing procedure for  layer $L_j$, where $j\in\mathbb{Z}^{+}\cup\{0\}$.


\vspace{1mm}\noindent\textbf{Preprocessing Procedure for $L_j$.}
Let $\sigma^{j}$ be a regular $k$-gon having diameter $2^j$. (i) Consider the partition $P_{\sigma^{j}}$ of the plane as described above. (ii) For every tile $s$ in $P_{\sigma^{j}}$, compute four sets of extreme points $V_{s,{\tau}}$.
(iii) For each $V_{s,{\tau}}$, construct a path graph $G_{s,\tau}$ over the set $V_{s,{\tau}}$ and then, compute a vertex ranking $c_{s,\tau}$ of the graph $G_{s,\tau}$ that uses $\lfloor\log_2(2|V_{s,\tau}|)\rfloor$ colors~\cite[Proposition~5]{EvenS14}.

\vspace{1mm}\noindent\textbf{Shrunk Set ${\cal K}_{\sigma}$ of a Regular $k$-gon $\sigma$.} Let $\sigma$ be a regular $k$-gon having diameter $w$ such that $2^j\leq w<2^{j+1}$, for some $j\in\mathbb{Z}^{+}\cup\{0\}$. If $w=2^j$, then we set  the shrunk set ${\cal K}_{\sigma}=\{\sigma\}$; otherwise,  ${\cal K}_{\sigma}$ is defined as follows. Let $\sigma_0'$ be a concentric scaled copy of $\sigma$ having diameter $w'=2^j$. Let $c_1,c_2,\ldots,c_k$ be $k$ corners of  $\sigma$. For each $j\in[k]$, we consider an exact copy $\sigma_{j}$  of $\sigma$, and  we shrink $\sigma_j$, keeping $c_j$ as the corner, such that $\sigma_j$ becomes a translated copy of $\sigma_0'$. Let $\sigma_j'$ be the obtained shrunk copy (see Figure~\ref{fig:bounded_scaled}).
We define ${\cal K}_{\sigma}=\{\sigma'_0,\sigma_1',\sigma_2',\ldots,\sigma_k'\}$.
For proof of the following lemma, we refer to Appendix~\ref{App_lem:shrunk}.

\begin{lemma}\label{lem:shrunk}
    $\sigma=\cup_{\sigma'\in{\cal K}_{\sigma}} \sigma'$.
\end{lemma}

\noindent\textbf{Description of Algorithm $\HHR$.}
We use $ES_{\sigma}$ to denote the algorithm $\ES$ for points in $\IR^2$, and translated copies of $\sigma$.
For each $j\in\mathbb{Z}^{+}\cup\{0\}$, let $L_j$ be the $j$th layer containing all objects with diameters between $[2^j,2^{j+1})$. 
The algorithm maintains a hitting set $\H=\cup_j\H_{j}$ for all objects that have been the part of input so far, where $\H_{j}$ is a hitting set maintained by our algorithm for objects of layer $L_j$ that have been the part of input so far. 
Upon the arrival of a new object $\sigma$, if it is already hit by $\H$, then do nothing. 
Otherwise, first, determine the layer $L_j$ in which the object $\sigma$ belongs. 
If $\sigma$ is the first arrived object from the layer $L_j$, then first proceed with the preprocessing procedure for $L_j$ as described above.
Now, irrespective of whether $\sigma$ is the first object of $L_{j}$ or not, consider the shrunk set ${\cal K}_{\sigma}$ of $\sigma$.
Each with equal probability, a sequence of the shrunk objects $\sigma_0',\sigma_1',\ldots,\sigma_k'$ in ${\cal K}_{\sigma}$ will be selected at random, and the shrunk objects arrive one at a time according to the selected sequence. 
Upon the arrival of each shrunk object of $\K_{\sigma}$, we apply $ES_{\sigma^j}$,
and update $\H_j$ by adding the hitting set for objects in ${\cal K}_{\sigma}$.
Note that to hit $\sigma$, one needs to hit only one of the objects in ${\cal K}_{\sigma}$, but to guarantee the competitive ratio, we hit all the objects in ${\cal K}_{\sigma}$.
Also note that if an (original/ shrunk) object does not contain any points of $\P$, it should be ignored. 

\begin{theorem}\label{Theo:Main_Bounded}
    The randomized algorithm $\HHR$ achieves a competitive ratio of {$4m_{\sigma}(k+1)^2\lfloor \log_2 2M \rfloor\lfloor \log_2 2n\rfloor$} for the hitting set problem of homothetic copies of a regular $k$-gon $\sigma$, for $(k\geq 4)$, having diameters in the range $[1, M]$ using $n$ points in $\IR^2$, {where $m_{\sigma}$ is a constant that depends on $k$}.
\end{theorem}

\begin{proof}
    Let $\I$ be a collection of homothetic objects presented to the online algorithm, and let $\I'\subseteq\I$ be the subset of objects that are unstabbed upon their arrivals. Let ${\cal K}=\cup_{\sigma\in{\I'}} {\cal K}_{\sigma}$ and let ${\cal K}'\subseteq{\cal K}$ be the subset of shrunk objects that are unstabbed upon their arrivals.
    For each $j\in\{0,1,\ldots,\lfloor \log M\rfloor\}$, let $\I'_j$ be the collection of all objects in $\I'$ having diameters in the range $\left[2^j, 2^{j+1}\right)$, and ${\cal K}'_j$ be the collection of all objects in ${\cal K}'$ having diameter $2^j$.
    Let $\A$ and $\OO$ represent the hitting sets returned by the online algorithm and the optimum offline algorithm, respectively, for the input sequence~$\I$. Let $x$ be a point in $\OO$. Let $\I'(x)\subseteq \I'$ (respectively, ${\cal K}'(x)\subseteq {\cal K}'$) be the  set of objects that contains the point $x$.
    For each $j\in\{0,1,\ldots,\lfloor \log M\rfloor\}$, let $\I'_j(x)=\I'(x)\cap\I'_j$ and  ${\cal K}'_j(x)={\cal K}'(x)\cap{\cal K}'_j$.
    Let $s$ be the tile containing $x$. Let $\A_{s,j}\subseteq \A$ be the set of points the online algorithm explicitly places on the tile $s$ to hit when some object in ${\cal K}'_j(x)$ arrives. 
    Note that ${\cal K}'_{j}(x)$ is the collection of all unstabbed shrunk objects containing the point $x$, and our algorithm places a hitting point for each unstabbed object. Thus, it is easy to observe that $|{\cal K}'_j(x)|\leq |\A_{s,j}|$.
    Due to Lemma~\ref{lem:distinct}, for  objects in ${\cal K}'_j(x)\cap{\cal D}_{s,{\tau}}$ our algorithm places distinct-colored extreme points from $V_{s,{\tau}}$, where $\tau \in [4]$. Due to~\cite[Proposition~5]{EvenS14}, the vertex ranking for $V_{s,\tau}$ uses at most $\lfloor \log_2(2|V_{s,{\tau}}|)\rfloor$ colors. Thus, $|\A_{s,j}| \leq 4\lfloor \log_2(2|V_{s,{\tau}}|)\rfloor$. 
    Since $|{\cal K}'_{j}(x)|\leq |\A_{s,j}|$ and $|\A_{s,j}| \leq 4\lfloor \log_2(2|V_{s,{\tau}}|)\rfloor$, we have $|{\cal K}'_j(x)|\leq 4\lfloor \log_2(2|V_{s,{\tau}}|)\rfloor\leq 4\lfloor \log_2 2n\rfloor$.


    Let $\A(X)\subseteq \A$ be the set of points the online algorithm explicitly places on all the tiles to hit when some object in $X\subseteq\I'$ arrives.
    Now, $\mathbb{E}(|\A(\I'(x))|) \leq \sum_{j=0}^{\lfloor \log_2 M \rfloor} \mathbb{E}(|\A(\I_j'(x))|)$.
    Since each object in $\I_j'(x)$ contains $(k+1)$ shrunk objects, and for each shrunk object, the algorithm places at most $m_{\sigma}$ number of points, we have, $\mathbb{E}(|\A(\I_j'(x))|) \leq m_{\sigma}(k+1)\mathbb{E}(|\I_j'(x)|)$.
    Now, we bound the cardinality of $\mathbb{E}(|\I_j'(x)|)$. Due to Lemma~\ref{lem:shrunk} and Markov's inequality, we have that $\mathbb{E}(|\I_j'(x)|) \leq (k+1)\mathbb{E}(|{\cal K}_j'(x)|)$. As a result, $\mathbb{E}(|\A(\I'(x))|) \leq \sum_{j=0}^{\lfloor \log_2 M \rfloor} m_{\sigma}(k+1)^2 \mathbb{E}|{\cal K}_j'(x)| \leq 4m_{\sigma}(k+1)^2\lfloor \log_2 2M \rfloor\lfloor \log_2 2n\rfloor$ (last inequality is due to $|{\cal K}'_j(x)| \leq 4\lfloor \log_2 2n\rfloor$).
    As  a result, we have $\mathbb{E}(|\A|)\leq \sum_{x\in\OO} \mathbb{E}{|\A(\I'(x))|} \leq \left(4m_{\sigma}(k+1)^2\lfloor \log_2 2M \rfloor\lfloor \log_2 2n\rfloor\right)|\OO|$. Hence, the theorem follows.
\end{proof}




\section{Conclusion}\label{conclusion}
In this paper, we have studied the two variants of the online hitting set problem, depending on the point set $\P$. In the first variant, the point set $\P$ is the entire $\IZ^d$. 
For hitting similarly sized homothetic hypercubes in $\mathbb{R}^d$, we first show that no algorithm, be it randomized or deterministic, can have a competitive ratio better than $\Omega(d\log M)$. Then, for the same problem, we propose a randomized algorithm with a competitive ratio of~$O(d^2\log M)$. Next, for hitting similarly sized fat objects in $\IR^d$, we present a deterministic algorithm having a competitive ratio of at most~$\lfloor\frac{2}{\alpha}+2\rfloor^d\left(\lfloor\log_2 M\rfloor+1\right)$.
Though our results improve the existing upper and lower bounds for the problem, there is room for improvement.
We state the following open problems.
\begin{enumerate}
    \item Is there a lower bound on the competitive ratio for hitting similarly sized homothetic hypercubes that matches the upper bound of the problem? Is there any algorithm for hitting similarly sized homothetic hypercubes, having side lengths in the range $[1, M]$ with a competitive ratio of $O(d\log M)$?
    \item There is a huge gap between the lower and the upper bounds for hitting similarly sized fat objects in $\IR^d$. We propose bridging this gap further as a future direction of research. 
\end{enumerate}

In the second variant, the point set $\P\subset\IR^2$ such that $|\P|=n$. For hitting similarly sized homothetic regular $k$-gons, we obtain a randomized algorithm with a competitive ratio of~$O(\log n\log M)$, partially answering an open question for squares~\cite{AlefkhaniKM23, KhanLRSW23}.  We state the following open problems.
\begin{enumerate}
    \item Design an $O(\log n)$ competitive algorithm for hitting arbitrary disks in $\IR^2$ ($d=2$) using points from $\P\subset\IR^2$  such that $|\P|=n$.
    \item {Design an $O(\log n)$ competitive algorithm for hitting similarly sized arbitrary oriented regular $k$-gons in $\IR^2$ using points from $\P\subset\IR^2$  such that $|\P|=n$.} 
\end{enumerate}

\newpage
\bibliography{references}

\newpage
\appendix
\section*{Appendix}

\section{Angle Property of Regular $k$-gons}\label{sec: Structural Property of Regular k-gons}
In this section, we prove the angle property of regular $k$-gons (Lemma~\ref{lemma:regular}) that is needed to establish the competitive ratio of algorithms $\HHR$ and $\ES$. 

\vspace{1mm}\noindent
$\blacktriangleright$
\textbf{Lemma~\ref{lemma:regular}}~\textsf{(Angle Property)}\textbf{.}~\textit{Let $\sigma$ and $\sigma'$ be two distinct translated copies of a regular $k$-gon such that their interior has a non-empty intersection. Let $p$ be any point in $\sigma\cap \sigma'$. Let $x$ and $y$ be any two points in $\partial \sigma\cap \partial \sigma'$ such that $x$ and $y$ belong to two different connected components of $\partial \sigma\cap \partial \sigma'$. Then, the angle $\angle xpy$ of the triangle $\triangle xpy$ is at least~$\frac{\pi}{4}$.}

\vspace{1mm}
Before proving the above lemma, we describe some notations. A {regular $k$-gon} is said to be \emph{odd} (respectively, \emph{even}) when $k$ is odd (respectively, even). A \emph{unit regular $k$-gon} is a regular $k$-gon in which the smallest disk enclosing it is of unit radius. Let $\sigma=a_0a_1\ldots a_{k-1}$ be a regular $k$-gon.  For any $i\in[k-1]\cup\{0\}$, the \emph{opposite vertex/vertices} of $a_i$ is/are defined as $a_{\left(i+\lceil\frac{k}{2}\rceil\right)\text{mod}{(k)}} \text{ and }  a_{\left(i+\lfloor\frac{k}{2}\rfloor\right)\text{mod}{(k)}}$ $\Big($respectively, $a_{\left(i+\frac{k}{2}\right)\text{mod}{(k)}}\Big)$, when $k$ is odd (respectively, even). Let  $a_j$ be an opposite vertex of $a_i$. Then, the line segment $\overline{a_ia_j}$  is known as a \emph{principal diagonal} of $\sigma$. Note that a regular odd (respectively, even) $k$-gon has $k$ (respectively,  $\frac{k}{2}$) principal diagonals.

\begin{observation}\label{Obs:relation}
    Let $r_{out}$ (respectively, $r_{in}$) be the radius of the smallest circumscribed (respectively, largest inscribed) disk of (respectively, in) a regular $k$-gon $\sigma$, and let $s$ be the length of each side of $\sigma$.  Then,  
    $r_{in}=r_{out}\cos{\left(\frac{\pi}{k}\right)}$, and
     $s=2r_{out}\sin{\left(\frac{\pi}{k}\right)}$.
\end{observation}
\begin{proof}
    Let $c$ be the center of $\sigma$, and let $b$ be the midpoint of an edge of $\sigma$ (see Figure~\ref{fig:Obs_1}). Note that  $\overline{ab}=\frac{s}{2}$. It is easy to observe that $\triangle  abc$ is a right-angled triangle. As a result, $r_{in}=r_{out}\cos(\frac{\pi}{k})$ and $\frac{s}{2}=r_{out}\sin{(\frac{\pi}{k})}$.
  \end{proof}

\begin{figure}[htbp]
    \centering
    \begin{subfigure}[b]{0.30\textwidth}
        \centering
        \includegraphics[page=1, width=27 mm] {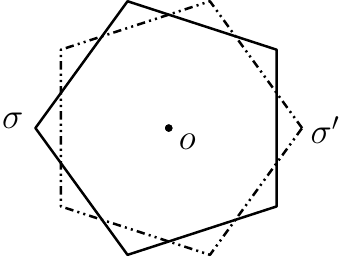}
        \caption{}
        \label{fig:ref}
    \end{subfigure}
    \begin{subfigure}[b]{0.30\textwidth}
        \centering
        \includegraphics[page=2, width=20 mm]{structure_k-gons.pdf}
        \caption{ }
        \label{fig:Obs_1}
    \end{subfigure}
    \begin{subfigure}[b]{0.30\textwidth}
        \centering
        \includegraphics[page=3, width=25 mm]{structure_k-gons.pdf}
        \caption{ }
        \label{fig:case_1.1}
    \end{subfigure}
    \caption{(a) The continuous black-colored regular pentagon is $\sigma$ and the dashed-dotted-dashed black-colored regular pentagon $-\sigma$ is the reflection of $\sigma$ through the center $o\in\sigma$. (b) Illustration of Observation~\ref{Obs:relation}. (c) Illustration of Observation~\ref{obs:0}.}
    \label{fig:notation}
\end{figure}



\begin{observation}\label{obs:0}
    Let $p$ be any point in a square $\square abcd$. Then, $\angle apb \geq \frac{\pi}{4}$.
\end{observation}

\begin{proof}
    First, consider the case when $p$ lies on edge  $L=\overline{cd}$.
    It is easy to observe that the angle $\angle apb$ achieves its minimum value of $\frac{\pi}{4}$ when $p$ coincides with either $c$ or $d$. Thus, the angle $\angle apb \geq \frac{\pi}{4}$, whenever $p$ is on edge $L$. Now, consider the case when the point $p$ is any point in the square $\square abcd$ (see Figure~\ref{fig:case_1.1}). Let $L'$ be a line perpendicular to $L$ and pass through the point $p$. Assume that  $L$ and $L'$ intersect at a point $p'$. Consider the following two triangles: $\triangle apb$ and $\triangle ap'b$. Note that $\angle pab$ and $\angle pba$ are at most equal to $\angle p'ab$ and $\angle p'ba$, respectively. Due to the angle sum property of triangles, we have $\angle~apb~\geq\angle~ap'b~\geq~\frac{\pi}{4}$.
\end{proof}

Now, we will prove the angle property of regular $k$-gon.


\noindent
\textbf{Proof of Lemma~\ref{lemma:regular}.}
    Let $\sigma$ and $\sigma'$ be two translated copies of a regular  $k$-gon. If $p\in \sigma\cap \sigma'$ is a point on the line segment $\overline{xy}$, we have nothing to prove. Thus, w.l.o.g., let us assume that $p$ is to the right of $\overline{xy}$. Let $R$ be the common intersection region of $\sigma\cap \sigma'$ lying to the right of the line segment $\overline{xy}$.
    Let $V=\{v \text{ is a vertex of }\sigma\ |\ v \text{ does not belongs to }\sigma'\}$. Let $D$ be the collection of principal diagonals of $\sigma$ such that both endpoints of the principal diagonal are in $V$.  We have two cases depending on whether $D$ is empty or not.

    \vspace{1mm}\noindent\textbf{Case~1.} $D$ is non-empty. Let $\overline{ab}\in D$. Also, let $\square abcd$ be a square that lies to the right side of $\overline{ab}$ (see Figure~\ref{fig:lemma:odd}). The principal diagonal is the largest line segment inside the regular $k$-gon. Thus, the square $\square abcd$ contains $R$. Now, notice that $\angle xpy$ is at least as big as $\angle apb$. Due to Observation~\ref{obs:0}, we have, $\angle apb \geq \frac{\pi}{4}$. Thus, $\angle xpy \geq \frac{\pi}{4}$.

\begin{figure}[htbp]
    \centering
    \begin{subfigure}[b]{0.49\textwidth}
         \centering
        \includegraphics[page=5, width=50 mm]{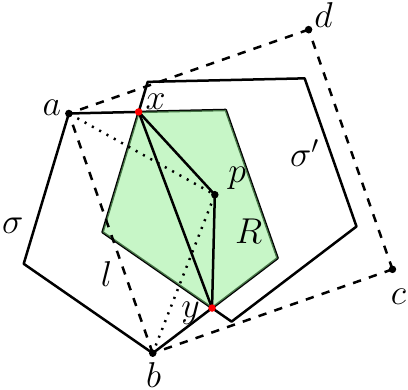}
         \caption{}
         \label{fig:lemma:odd}
    \end{subfigure}
    \begin{subfigure}[b]{0.49\textwidth}
        \centering
        \includegraphics[page=6, width=50 mm]{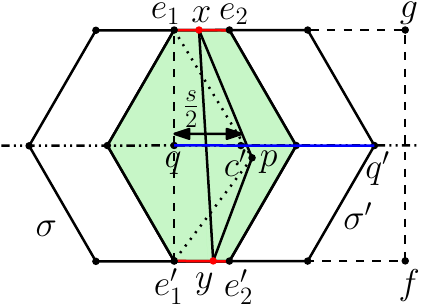}
        \caption{}
        \label{fig:evengon_case2}
    \end{subfigure}
    \caption{
    Illustration of Lemma~\ref{lemma:regular}. (a) Case 1 and (b) Case 2. The red regions denote the two connected components of $\partial\sigma\cap\partial\sigma'$.}
    \label{fig:case_6}
\end{figure}

    \noindent\textbf{Case~2.} 
    $D$ is empty. Since $D=\emptyset$, the common intersection region $\sigma\cap\sigma'$ has at least  $\lceil \frac{k}{2}\rceil$ consecutive vertices of $\sigma$. For regular odd $k$-gons, this is only possible when $\sigma$ and $\sigma'$ coincide. Since $\sigma$ and $\sigma'$ are distinct translated copies, the set $D$ is always non-empty for regular odd $k$-gons. When $\sigma$ and $\sigma'$ are translated copies of a regular even $k$-gon and  $\sigma\cap\sigma'$ has at least $\lceil \frac{k}{2}\rceil$ consecutive vertices of $\sigma$, then  $\partial\sigma\cap\partial\sigma'$ consists of two parallel edges, say, $e=\overline{e_1e_2}$ and $e'=\overline{e_1'e_2'}$ (see Figure~\ref{fig:evengon_case2}). Let $x\in e$ and $y\in e'$ be two arbitrary points. Let $\square {e_1 e_1'f g}$ be a square containing the line segment $\overline{xy}$. 
    Let $L$ be a line perpendicular to the line segment $\overline{e_1e_1'}$ and pass through the center $c'$ of $\sigma'$. Let   $L$  and $\overline{e_1e_1'}$ intersect at a point $q$. Let $q'$ be the point of intersection with  $L$ and $\sigma'$ lying to the right side of $\overline{xy}$. 
    It is easy to observe that the length of the line segment $\overline{qc'}$ is $\frac{s}{2}$, where $s$ is the side length of $\sigma$.  The length of the line segment $\overline{c'q'}$ is at most $r_{out}$, where $r_{out}$  is the radius of the smallest circumscribed disk of $\sigma$. Consequently, the length of the line segment $\overline{qq'}$ is at most $r_{out}+\frac{s}{2}$. Note that the length of $\overline{e_1e_1'}$ is $2r_{in}$, where $r_{in}$ is the radius of the largest inscribed disk of $\sigma$.
    Due to Observation~\ref{Obs:relation}, the length of $\overline{e_1e_1'}$ is $2r_{out}\cos (\pi/k)\geq2r_{out}\cos (\pi/6)=\sqrt{3}r_{out}$ and the length of $\overline{qq'}$ is at most $r_{out}(1+\sin (\pi/k))\leq r_{out}(1+\sin (\pi/6))= 1.5r_{out}$. Observe that the length of $\overline{e_1e_1'}$ is greater than the length of $\overline{qq'}$.  Thus, the square $\square e_1 e_1'fg$ contains $R$.
    Notice that $\angle xpy$ is at least as big as $\angle e_1pe_1'$. Due to Observation~\ref{obs:0}, we have $\angle e_1pe_1' \geq \frac{\pi}{4}$. Thus, $\angle xpy \geq~\frac{\pi}{4}$. Hence, Lemma~\ref{lemma:regular} holds for unit regular $k$-gons with  $k\geq 4$.

\section{Missing Proofs}
In this section, we give the proofs of the technical theorems and lemmas that are missing in the main body of the paper.

\subsection{Missing Proof of Lemma~\ref{obs: homothetic child node relation}}\label{App_obs: homothetic child node relation}

\noindent
$\blacktriangleright$ 
\textbf{Lemma~\ref{obs: homothetic child node relation}.} 
        \textit{Let $N$, $N'$ and $N''$ be three nodes of the tree $\T$ such that $N'$ and $N''$ are the two children of $N$. Also, let $\sigma'$ and $\sigma''$ be two hypercubes corresponding to the nodes $N'$ and $N''$, respectively. Then, $\sigma'$ and $\sigma''$ are interior disjoint.}
        
    \begin{proof}
        Let $N=N^{k,l}_{i,j}$ for some $k\in[\log_4 M]$, $l\in [2^{d(k-1)}]$, $i\in [d]$ and $j\in [2^{i-1}]$. First, consider the case when $N=N^{k,l}_{i,j}$ is an internal node of $T^{k,l}_{B}$, i.e., 
       $k\in[\log_4 M]$, $l\in [2^{d(k-1)}]$, $i\in [d-1]$ and $j\in [2^{i-1}]$. Then, we have $N'=N^{k,l}_{i+1,2j-1}$ and $N''=N^{k,l}_{i+1,2j}$, and the two hypercubes $\sigma'$ and $\sigma''$ are centered at points $c^{k,l}_{i+1,2j-1}$ and $c^{k,l}_{i+1,2j}$, respectively, having a width $\frac{M}{4^{k-1}}$. Recall that $c^{k,l}_{i+1,2j-1}=c^{k,l}_{i,j}+\frac{M}{4^{k-1}}{\bf e}_{i}$ and $c^{k,l}_{i+1,2j}=c^{k,l}_{i,j}-\frac{M}{4^{k-1}}{\bf e}_{i}$. Thus, the centers of the two hypercubes $\sigma'$ and $\sigma''$ only differ in the $i$th coordinate, and the distance between them is exactly $\frac{2M}{4^{k-1}}$. Since both $\sigma'$ and $\sigma''$ are of width $\frac{M}{4^{k-1}}$,  they are interior disjoint.

        Now, consider the case when $N=N^{k,l}_{d,j}$ is a leaf node of $T^{k,l}_{B}$, i.e., $k\in[\log_4 M-1]$, $l\in [2^{d(k-1)}]$ and $j\in [2^{d-1}]$. Then, we have $N'=N^{k+1,2^d(l-1)+2j-1}_{1,1}$ and $N''=N^{k+1,2^d(l-1)+2j}_{1,1}$, and  the two hypercubes $\sigma'$ and $\sigma''$ are centered at points $c^{k+1,2^d(l-1)+2j-1}_{1,1}$ and $c^{k+1,2^d(l-1)+2j}_{1,1}$, respectively, having a width $\frac{M}{4^{k}}$. Recall that $c^{k+1,2^d(l-1)+2j-1}_{1,1}=\frac{c^{k,l}_{1,1}+c^{k,l}_{d,j}}{2}+\frac{M}{2^{2k-1}}{\bf e}_{d}$ and $c^{k+1,2^d(l-1)+2j}_{1,1}=\frac{c^{k,l}_{1,1}+c^{k,l}_{d,j}}{2}-\frac{M}{2^{2k-1}}{\bf e}_{d}$. Thus, the centers of the two hypercubes $\sigma'$ and $\sigma''$ only differ in the $d$th coordinate, and the distance between them is exactly $\frac{M}{2^{2k-2}}=\frac{4M}{4^{k}}$. Since both $\sigma'$ and $\sigma''$ are of width $\frac{M}{4^{k}}$, they are interior disjoint.
    \end{proof}

\subsection{Missing Proof of Lemma~\ref{claim:lb_hyper}}\label{App_claim:lb_hyper}

\noindent
$\blacktriangleright$ \textbf{Lemma~\ref{claim:lb_hyper}.} \textit{Let $\T^{k,l}_B$ be a tree for some $k\in[\log_4 M-1]$ and $l\in[2^{d(k-1)}]$. Let ${\cal J}=(\sigma_u)_{u=1}^{d}$ be a sequence of hypercubes having a width $w=\frac{M}{4^{k-1}}$ corresponding to nodes along a path from the root $N_{1,1}^{k,l}$ to a leaf $N^{k,l}_{d,j}$ of the tree $\T^{k,l}_B$ for some $j\in[2^{d-1}]$. Then, we have the following:}
\begin{enumerate}
    \item[\textit{(i)}] \textit{The common intersection region $Q=\cap_{u=1}^{d} \sigma_u$ contains two interior disjoint hypercubes $H_1$ and $H_2$ having side length $\frac{w}{2}$.}
    
    \item[\textit{(ii)}] \textit{Let $\T_B^{k+1,l'}$ be a tree rooted at a child of the leaf node $N^{k,l}_{d,j}$ of the tree $\T_B^{k,l}$, where $l'\in \{2^d(l-1)+2j-1, 2^d(l-1)+2j\}$. Let ${\Pi}=(\sigma'_u)_{u=1}^{d}$ be a sequence of hypercubes corresponding to a path from the root to a leaf of the tree $\T_B^{k+1,l'}$. Then, the union of the hypercubes $\cup_{u=1}^{d} \sigma'_u$ is contained in $\cap_{u=1}^{d} \sigma_u$.}
\end{enumerate}
    

\begin{proof}
    First, from the construction of the hypercubes, recall that the center of the hypercube $\sigma_1$ corresponding to the node $N^{k,l}_{1,1}$ is $c^{k,l}_{1,1}$. Now, consider the next hypercube $\sigma_2$. Recall that the center of $\sigma_2$ is a translated copy of the center of $\sigma_1$ differing only in the $1$st coordinate. Specifically, this translation is by $w{\bf e}_1$ (or $-w{\bf e}_1$) for left (or right) child. Hence, $c^{k,l}_{2,j'}=c^{k,l}_{1,1}+sign(\sigma_2)w{\bf e}_{1}$, where $w=\frac{M}{4^{k-1}}$ and the value of $sign(\sigma_2)=1$ if $\sigma_{2}$ is the left child of $\sigma_{1}$, and $-1$ if $\sigma_{2}$ is the right child of $\sigma_{1}$. As a result of this, and $\sigma_1$ and $\sigma_2$ are both of width $w$, we also have $\sigma_2=\sigma_{1}+sign(\sigma_2)w{\bf e}_{1}$. Similarly, for each $u\in[d]\setminus\{1,2\}$, we have $c^{k,l}_{u,j'}=c^{k,l}_{u-1, \lfloor\frac{j'}{2}\rfloor}+sign(\sigma_u)w{\bf e}_{u-1}$ and $\sigma_u=\sigma_{u-1}+sign(\sigma_u)w{\bf e}_{u-1}$, where the value of $sign(\sigma_u)=1$ if $\sigma_{u}$ is the left child of $\sigma_{u-1}$, and $-1$ if $\sigma_{u}$ is the right child of $\sigma_{u-1}$. Therefore, for each $u\in[d]\setminus\{1\}$ and $j'\in [2^{u-1}]$, we have

    \vspace{-3mm}
    \begin{equation}\label{eq: c}
        c^{k,l}_{u,j'}=c^{k,l}_{1, 1}+\sum_{v=2}^{u} sign(\sigma_v)w{\bf e}_{v-1}
    \end{equation}
    and 
    \begin{equation}\label{eq: sigma}
        \sigma_u=\sigma_{1}+\sum_{v=2}^{u} sign(\sigma_v)w{\bf e}_{v-1}.
    \end{equation}

\begin{enumerate}
    \item[(i)] 
    Since $\sigma_2=\sigma_1 +sign(\sigma_2)w{\bf e}_1$, we have that $\sigma_1\cap\sigma_2$ is a hyperrectangle centered at the point $c^{k,l}_{1,1}+sign(\sigma_2)\frac{w}{2}{\bf e}_{1}$ whose side corresponding to the first coordinate is of length $w$, and all the remaining $(d-1)$ sides are of length $2w$. Since $\sigma_3=\sigma_2+sign(\sigma_3) w{\bf e}_2=\sigma_1+\sum_{v=2}^{3}sign(\sigma_v) w{\bf e}_{v-1}$ (due to Equation~\ref{eq: sigma}), thus $\sigma_1\cap\sigma_2\cap\sigma_3$ is a hyperrectangle centered at the point $c^{k,l}_{1,1}+\sum_{v=2}^{3}sign(\sigma_v)\frac{w}{2}{\bf e}_{v-1}$ whose sides corresponding to the first two coordinates are of length $w$ each, and the remaining $(d-2)$ sides are of length $2w$. Since for any $u\in[d]\setminus\{1\}$, the hypercube $\sigma_u=\sigma_{u-1}+sign(\sigma_u) w{\bf e}_{u-1}=\sigma_{1}+\sum_{v=2}^{u} sign(\sigma_v)w{\bf e}_{v-1}$ (due to Equation~\ref{eq: sigma}), we have $Q=\cap_{u=1}^{d} \sigma_{u}$ as a hyperrectangle centered at a point $c^{k,l}_{1, 1}+\sum_{v=2}^{u} sign(\sigma_v)\frac{w}{2}{\bf e}_{v-1}$ whose sides corresponding to the first $(d-1)$ coordinates are of length $w$ each, and the remaining side is of length  $2w$.
    
    \quad As a result, the common intersection region $Q$ will contain two interior disjoint hypercubes $H_1$ and $H_2$ of width $\frac{w}{2}$ centered at $p_1$ and $p_2$, respectively, where
    \begin{align*}
        p_1&=c^{k,l}_{1,1}+\sum_{u=2}^{d} sign(\sigma_u)\frac{w}{2} {\bf e}_{u-1}+\frac{w}{2}{\bf e}_{d}\\
        &=\frac{c^{k,l}_{1,1}}{2}+\left(\frac{c^{k,l}_{1,1}}{2}+\frac{1}{2}\sum_{u=2}^{d} sign(\sigma_u)w {\bf e}_{u-1}\right)+\frac{w}{2}{\bf e}_{d}\\
        &=\frac{c^{k,l}_{1,1}}{2}+\frac{c^{k,l}_{d,j}}{2}+\frac{w}{2}{\bf e}_{d}
        =\frac{c^{k,l}_{1,1}+c^{k,l}_{d,j}}{2}+\frac{w}{2}{\bf e}_{d},
        \tag{due to Equation~\ref{eq: c}}
    \end{align*}
 and 

 \begin{align*}
        p_2&=c^{k,l}_{1,1}+\sum_{u=2}^{d} sign(\sigma_u)\frac{w}{2} {\bf e}_{u-1}-\frac{w}{2}{\bf e}_{d}\\
        &=\frac{c^{k,l}_{1,1}}{2}+\left(\frac{c^{k,l}_{1,1}}{2}+\frac{1}{2}\sum_{u=2}^{d} sign(\sigma_u)w {\bf e}_{u-1}\right)-\frac{w}{2}{\bf e}_{d}\\
        &=\frac{c^{k,l}_{1,1}}{2}+\frac{c^{k,l}_{d,j}}{2}-\frac{w}{2}{\bf e}_{d}
        =\frac{c^{k,l}_{1,1}+c^{k,l}_{d,j}}{2}-\frac{w}{2}{\bf e}_{d}.
        \tag{due to Equation~\ref{eq: c}}
    \end{align*}

    \item[(ii)] 
    From the construction of the centers of the hypercubes, we have that $\sigma'_1$ is centered at the point $\frac{c^{k,l}_{1,1}+c^{k,l}_{d,j}}{2}+sign(\sigma'_1)\frac{w}{2}{\bf e}_{d}$, where $sign(\sigma'_1)=1$ if $\sigma'_{1}$ is the left child of $\sigma_{d}$, and $sign(\sigma'_1)=-1$ if $\sigma'_{1}$ is the right child of $\sigma_{d}$.
    W.l.o.g., assume that $\sigma'_1$ is a hypercube corresponding to the left child of the node $N^{k,l}_{d,j}$. The other case, i.e., $\sigma'_1$ is a hypercube corresponding to the right child of $N^{k,l}_{d,j}$, can be done in a similar way.
    From (i), we have that $\cap_{u=1}^{d} \sigma_u$ contains a hypercube $H_1$ having width $\frac{w}{2}$ centered at the point $\frac{c^{k,l}_{1,1}+c^{k,l}_{d,j}}{2}+\frac{w}{2}{\bf e}_{d}$.
    Notice that $\sigma'_1$ and $H_1$ are concentric hypercubes having widths $\frac{w}{4}$ and $\frac{w}{2}$, respectively. Hence, $\sigma'_1\subset H_1$.
    In a similar way as we have Equation~\ref{eq: sigma}, for each $u\in[d]\setminus\{1\}$, we also have $\sigma'_u=\sigma'_{1}+\sum_{v=2}^{u} sign(\sigma'_v)\frac{w}{4}{\bf e}_{v-1}$, where $sign(\sigma'_v)=1$ if $\sigma'_{v}$ is the left child of $\sigma'_{v-1}$, and $sign(\sigma'_v)=-1$ if $\sigma'_{v}$ is the right child of $\sigma'_{v-1}$.
    Thus, the centers of the hypercubes $\sigma'_1$ and $\sigma'_u$ only differ in the first $(u-1)$th coordinates, and the difference is exactly $\frac{w}{4}$ at each coordinate for $u\in[d]\setminus\{1\}$. As a result of this, and $\sigma'_1$ and $H_1$ are concentric hypercubes, and the hypercube $\sigma'_u$ is of width $\frac{w}{4}$, we have that $\sigma'_u\subset H_1$ for each $u\in[d]\setminus\{1\}$.
    Therefore, we have $\cup_{u=1}^{d} \sigma'_u\subseteq H_1 \subseteq \cap_{u=1}^{d} \sigma_u$. Hence, the lemma follows.

\end{enumerate}
\end{proof}

\subsection{Missing Proof of Claim~\ref{clm_core}}\label{App_clm_core}
\noindent
$\vartriangleright$ \textsf{Claim~\ref{clm_core}.}\textit{
  $\Q_{s(k)}(\sigma)=\Q_{s(k)}(H)$.}

\begin{proof}
    First, we demonstrate $\Q_{s(k)}(\sigma)\subseteq\Q_{s(k)}(H)$.
Since the convex hull $C$ of $\sigma$ containing all the points of $\sigma$ from $\IZ_k$ is totally contained inside $H$, we have $\Q_{s(k)}(\sigma)\subseteq\Q_{s(k)}(H)$. Next, we prove $\Q_{s(k)}(H)\subseteq\Q_{s(k)}(\sigma)$. 
If the side length of $C$ is $2s(k)$, then that will also work as a side length of $H$. If the side length of $C$ is $s(k)$, we extend it by $\frac{s(k)}{2}$ in both directions to make it of length $2s(k)$, considering it as the side of $H$. Since we are extending it only by $\frac{s(k)}{2}$, it will not contain any new point from $\IZ_k$ apart from the points of $C$. Consequently,  we have $\Q_{s(k)}(H)\subseteq\Q_{s(k)}(\sigma)$. Hence, the claim follows.
\end{proof}

\subsection{Missing Proof of Claim~\ref{clm: r closest point}}\label{App_clm: r closest point}
\begin{claim}\label{clm: r closest point}
    The point $r$ is the closest point of $(2^{i+1}\mathbb{Z})^d$ from the centre $c$.
\end{claim}

\begin{proof}
    Let $r'(\neq r)$ be any point of $(2^{i+1}\mathbb{Z})^d$.
    Notice that for any $j\in[d]$, we have 
    \begin{align*}
    |r'(x_j)-c(x_j)|&=|r'(x_j)- z_j-f_j|\tag{$c(x_j)=z_j+f_j$, where $z_j\in 2^{i+1}\mathbb{Z}$ and $f_j\in\left[0,2^{i+1}\right)$}\\
        &\geq ||r'(x_j)-z_j|-|f_j||\tag{Reverse triangle inequality, i.e., $|a-b|\geq||a|-|b||$}\\
        &\geq|0-|f_j||=|f_j| \tag{Since $r'(x_j)\in(2^{i+1}\mathbb{Z})^d$, we have $|r'(x_j)-z_j|\geq 0$}\\
        &\geq|r(x_j)-c(x_j)|\tag{As per the construction of $r$, $|r(x_j)-c(x_j)|\leq|f_j|$}. 
    \end{align*}
    
    Since for each $j\in[d]$, we have $|r'(x_j)-c(x_j)|\geq |r(x_j)-c(x_j)|$,  the point $r$ is the closest point of $(2^{i+1}\mathbb{Z})^d$ from the center $c$. Hence, the claim follows.
\end{proof}

\subsection{Missing Proof of Lemma~\ref{lem_correct}}\label{App_lem_correct}
\noindent
$\blacktriangleright$
\textbf{Lemma~\ref{lem_correct}.}\textit{
The object $\sigma$ contains the point $r$ from $(2^{i+1}\mathbb{Z})^d$.}

\begin{proof}
   Due to Claim~\ref{clm:drc}, we have $d_{\infty}(r,c)\leq 2^{i}$. Thus, a hypercube centered at $c$ having side length $2^{i+1}$ contains the point  $r$ from $(2^{i+1}\mathbb{Z})^d$. Notice that the object $\sigma$ centered at $c$ having a width at least $2^i$ contains a concentric hypercube of side length at least $2^{i+1}$ inside it. 
    Hence, the lemma follows.    
\end{proof}

\subsection{Missing Proofs of Subsection~\ref{sec: algo-es}}\label{sec: appendix algo-es}
In this section, we primarily give the missing proofs of Subsection~\ref{sec: algo-es}. Before giving the proofs, we again describe the algorithm $\ES$ and some important results that are essential for those proofs.

\vspace{1mm}\noindent\emph{Preprocessing Step.}
As part of preprocessing, we do the following. (i) Consider a partition $P_{\sigma}$ of the plane. (ii) For every tile $s$ in $P_{\sigma}$ such that $s\cap \P\neq \emptyset$, compute four sets of extreme points $V_{s,{\tau}}$ for $\tau\in [4]$.
(iii) For each $V_{s,{\tau}}$, construct a path graph $G_{s,\tau}$ over the set $V_{s,{\tau}}$ and then, compute a vertex ranking $c_{s,\tau}$ of the graph $G_{s,\tau}$ that uses $\lfloor\log_2(2|V_{s,\tau}|)\rfloor$ colors~\cite[Proposition~5]{EvenS14}.


\vspace{1mm}Now, we will describe each of the three preprocessing steps in detail.

\vspace{1mm}\noindent\hspace{2mm} $\bullet$ \underline{Partitioning the Plane}: 
Let $\sigma$ be a regular $k$-gon. Depending on the shape and size of the object $\sigma$, we define a parameter $\ell_{\sigma}$. Assuming the diameter of $\sigma$ to be 2 units, we set the value of $\ell_{\sigma}$ as follows. If $\sigma$ is a square (respectively, regular $k$-gon with $k\geq 7$ and   regular $k$-gon with $k\in\{5,6\}$),  the value of $\ell_{\sigma}$ is $\frac{1}{2}$ (respectively, $\frac{1}{2}$ and $\frac{1}{4}$).
Similar to~\cite{ChenKS09}, we partition the plane using square tiles each of side length $\ell_{\sigma}$ such that each point $p\in\P$ lies in the interior of a tile. It is easy to achieve such a partition $P_{\sigma}$, since the set $\P$ is finite.
Let $m_{\sigma}$ be the maximum number of tiles of the partition $P_{\sigma}$ that any translated copy of $\sigma$ can intersect.

\vspace{2mm}\noindent
$\blacktriangleright$
\textbf{Observation~\ref{obs:contain}.}\textit{
The value of $m_{\sigma}\leq 25$ $($respectively, $23$ and $63)$, when  $\sigma$ is a  square $($respectively, regular $k$-gon with $k\geq 7$ and  regular $k$-gon with $k\in\{5,6\})$.}
 
\begin{figure}[htbp]
    \centering
    \begin{subfigure}[b]{0.45\textwidth}
        \centering
        \includegraphics[page=2, width=47 mm]{homothetic_kgons_appendix.pdf}
        \caption{}
        \label{fig:cone1}
    \end{subfigure}
    \begin{subfigure}[b]{0.45\textwidth}
        \centering
        \includegraphics[page=1, width=45 mm]{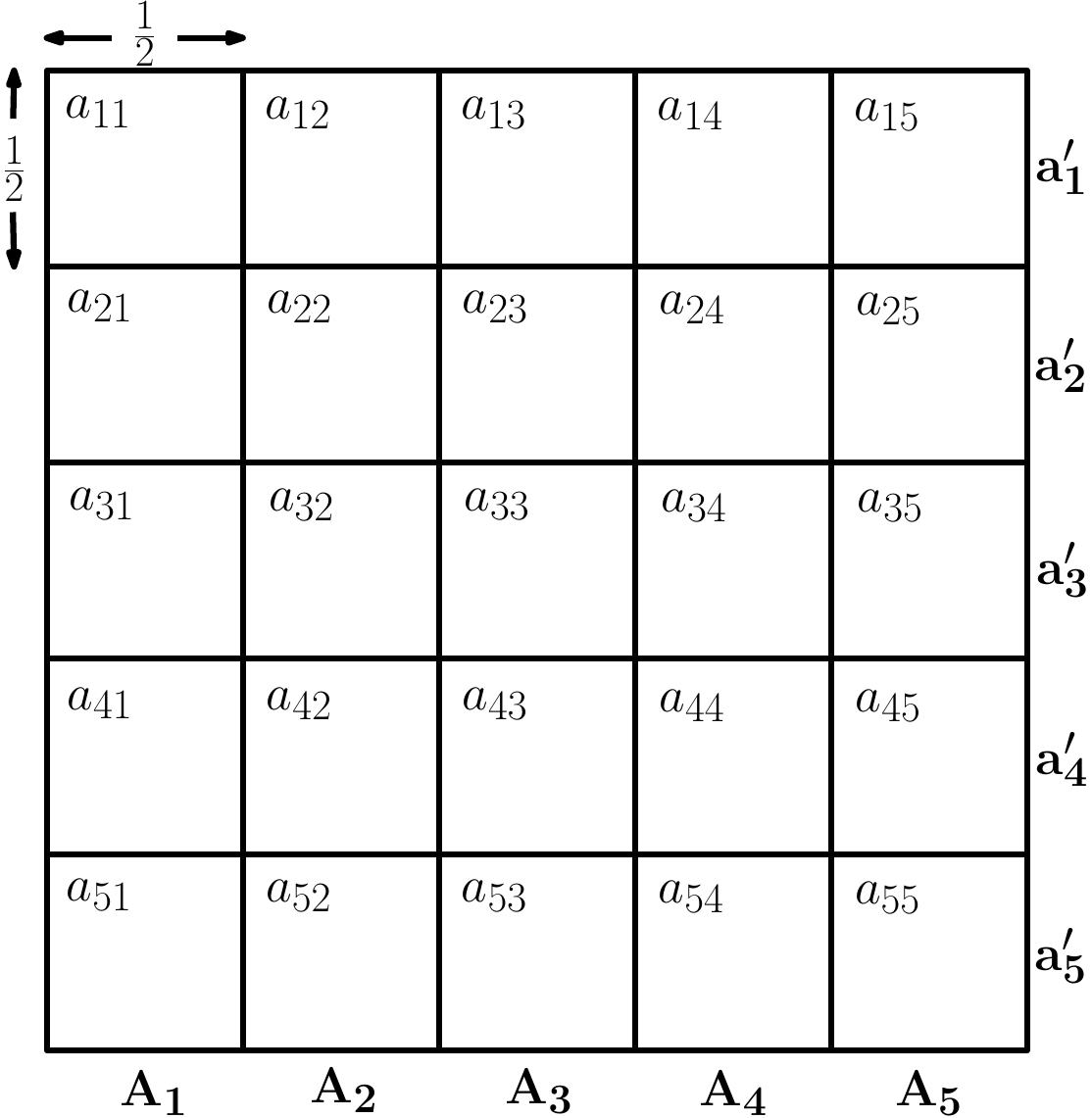}
        \caption{}
        \label{fig:disk}
    \end{subfigure}
    \caption{(a) Super-square $S$ of a tile  $s$ of the partition $P_{\sigma}$, where $\sigma$ is a unit square. (b) A $5\times 5$ grid such that the side length of each tile in the grid is $\frac{1}{2}$.}
    \label{fig:count}
\end{figure}

\begin{proof}
 Let ${\bf{A}}=\left(a_{ij}\right)_{t\times t}$ be a {$t\times t$} partition (see Figure~\ref{fig:disk}). Let $\bf{a_i}'$ and $\bf{A_j}$ denotes the $i$th row and $j$th column of $\bf{A}$. We use $a_{ij}$ to denote a tile in the intersection of the $i$th row and $j$th column, where $i,j\in[t]$. W.l.o.g., we assume that the object $\sigma$ intersects a tile in both $\bf{a_1}'$ and $\bf{A_1}$. 

When $\sigma$ is a unit square in $\IR^2$, then the side length of each tile in $\bf{A}$ is $\frac{1}{2}$.
  Note that the distance between any two points $x$ and $y$ such that $x\in \interior(\bf{A_1})$ (respectively, $\interior(\bf{a_1}')$) and $y\in \interior(\bf{A_6})$ (respectively, $\interior(\bf{a_6}')$) is strictly greater than 2. Thus, any unit square intersecting a tile in $\bf{A_1}$ (respectively, $\bf{a_1'}$), cannot intersect any tile in $\bf{A_6}$ (respectively, $\bf{a_6}'$) from the partition $\bf A$. As a result of this and the fact that the unit square $\sigma$ intersects a tile in both $\bf{a_1}'$ and $\bf{A_1}$, it can only intersect the tiles $a_{ij}$ for $i,j\in [5]$. Thus, the square $\sigma$ intersects at most twenty-five tiles from $\bf A$. 

    When $\sigma$ is a unit regular $k$-gon with $k\geq 7$ in $\IR^2$, then the side length of each tile in $\bf{A}$ is $\frac{1}{2}$. Any unit regular $k$-gon $\sigma(c)$ with fixed $k\geq 7$ is bounded in a disk of unit radius. Note that the distance between any two points $x$ and $y$ such that $x\in \interior(\bf{A_1})$ (respectively, $\interior(\bf{a_1}')$) and $y\in \interior(\bf{A_6})$ (respectively, $\interior(\bf{a_6}')$) is strictly greater than 2. Hence, any unit disk  intersecting a tile in $\bf{A_1}$ (respectively, $\bf{a_1'}$), cannot intersect any tile in $\bf{A_6}$ (respectively, $\bf{a_6}'$) from the partition $\bf A$. As a result of this and the fact that the unit disk $\sigma$ intersects a tile in both $\bf{a_1}'$ and $\bf{A_1}$, it can only intersect the tiles $a_{ij}$ for $i,j\in [5]$, i.e., at most twenty-five tiles from $\bf{A}$. Also, observe that the distance between any two points $x$ and $y$ such that $x\in a_{11}$ (respectively, $a_{15}$) and $y\in a_{55}$ (respectively, $a_{51}$) is at least~2.1. Thus, any unit disk intersecting tile $a_{11}$ (respectively, $a_{15}$) cannot intersect tile $a_{55}$ (respectively, $a_{51}$) or vice versa. As a result of this and the fact that the unit disk completely contains $\sigma$, we have that the regular $k$-gon $\sigma$ with $k\geq 7$ intersects at most twenty-three tiles from $\bf A$.

    When $\sigma$ is a  regular $k$-gon with $k\in\{5,6\}$, then the side length of each tile $a_{ij}$ in $\bf{A}$ is $\frac{1}{4}$. Any unit regular $k$-gon $\sigma(c)$ with fixed $k\geq 5$ is bounded in a unit disk. 
    With a similar argument as described in the previous paragraph, we can show that the unit disk cannot intersect more than sixty-two tiles of side length $\frac{1}{4}$ from $\bf A$. As a result of this and the fact that the unit disk completely contains $\sigma$, we have that the regular $k$-gon $\sigma$ with $k\in\{5,6\}$ intersects at most sixty-two from $\bf A$.
\end{proof}

\noindent\hspace{2mm} $\bullet$ \underline{Extreme Points}: 
Depending on the shape and size of the object $\sigma$, we define a parameter $t_{\sigma}$. Assuming the diameter of $\sigma$ as 2 units, we set the value of $t_{\sigma}$ as follows. If $\sigma$ is a square (respectively, regular $k$-gon with $k\geq 7$ and regular $k$-gon with $k\in\{5,6\}$)  the value of $t_{\sigma}$ 
is~$\frac{5}{2}$ $\left(\text{respectively, } \frac{5}{2} \text{ and } \frac{9}{4}\right)$.
For each square tile $s$ from the partition $P_{\sigma}$, we define a \emph{super-square} $S$ that is a square concentric with the tile $s$ whose side length is $t_{\sigma}$. Let ${\cal D}_s$ be the set of all translated copies of $\sigma$ with non-empty intersection with the tile $s$.
It is easily observed that $S$  contains centers of all objects in ${\cal D}_s$. 
We divide the super-square $S$ into four quadrants, each a square of side length $t_{\sigma}/2$.  For $\tau\in [4]$, let  $S_{\tau}$ denote a quadrant of $S$ and let $o_{\tau}$ denote its center (see Figure~\ref{fig:cone1}).

\vspace{2mm}\noindent$\blacktriangleright$ \textbf{Observation~\ref{obs:inter}.}\textit{ If $\sigma' \in{\cal D}_s$, then  $\sigma'\cap\{o_1,o_2,o_3,o_4\}\neq\emptyset$.}
\begin{proof}
    Note that for any object $\sigma'(x)\in{\cal D}_s$, the center $x$ lies in the super-square $S$. Let us assume that $x$ lies in $S_{\tau}$, for some ${\tau}\in [4]$. Since $\sigma'(x)$ contains a disk of radius $\cos\left(\frac{\pi}{k}\right)$,  it is sufficient to prove that the maximum (Euclidean) distance from any point in $S_{\tau}$ to $o_{\tau}$ is at most $\cos\left(\frac{\pi}{k}\right)$ units. Let $c$ be the center of the tile $s$. Since $c$ is one of the farthest points from $o_{\tau}$, it is enough to show that $d(o_{\tau},c) \leq \cos\left(\frac{\pi}{k}\right)$.

    {First, consider the case when $\sigma$ is a unit square, i.e., a square of side length $\sqrt{2}$, while $S_{\tau}(o_{\tau})$ is a square of side length $\frac{5}{4}$. Notice that the $L_{\infty}$ distance between $c$ and $o_{\tau}$ is $\frac{5}{8}<\frac{1}{\sqrt{2}}$.}
    Next, consider the case when $\sigma$ is a regular $k$-gon with fixed $k\geq 7$. Note that the $d(c,o_{\tau})=\frac{5\sqrt{2}}{8}<\cos\left(\frac{\pi}{k}\right)$, when $k\geq 7$.  
    Finally,  consider the case when $\sigma$ is  a unit regular $k$-gon with fixed $k\in\{5,6\}$. Note that the $d(c,o_{\tau})<0.797<\cos\left(\frac{\pi}{k}\right)$ for both $k=5,6$.
\end{proof}

Let $K_s^{\tau}$ be the convex cone with apex $o_{\tau}$ spanned by the tile~$s$ (see Figure~\ref{fig:cone1}). Here, we give an upper bound of the opening angle of $K_s^{\tau}$. 

\begin{lemma}\label{lemma:cone}
    The opening angle of the cone $K_s^{\tau}$ is strictly less than $\frac{\pi}{2}$. In particular, when $\sigma$ is a  regular $k$-gon with  $k\geq 4$, the opening angle of the cone $K_s^{\tau}$ is strictly less than $\frac{\pi}{4}$.
\end{lemma}

\begin{proof}
    Let us assume that  $K_s^{\tau}$ touches the tile $s$ at two points $x$ and $y$. It is easy to argue that the opening angle of the cone $K_s^{\tau}$ is strictly less than  $\frac{\pi}{2}$. To see this, consider two triangles  $\triangle xyz$ and $\triangle xyo_{\tau}$ (see  Figure~\ref{fig:imp}); note that the angle $\angle xzy=\frac{\pi}{2}$ and $\angle xo_{\tau}y< \angle xzy$. Thus, the angle $\angle xo_{\tau}y<\frac{\pi}{2}$.
 
    Now, consider the specific case of regular $k$-gon. Let us consider the triangle $\triangle xyo_{\tau}$. By applying the law of cosines, we get

    \begin{equation}\label{eqn_1}
        \angle xo_{\tau}y=\cos^{-1}{\left(\frac{d^2(o_{\tau},x)+d^2(o_{\tau},y)-d^2(y,x)}{2d(o_{\tau},x)d(o_{\tau},y)}\right)}.
    \end{equation}

    First, consider the case when $\sigma$ is a regular 5-gon (respectively,  regular 6-gon), then note that $d(o_{\tau},x)$= $d(o_{\tau},y)$$\approx0.999$ (respectively, $d(o_{\tau},x)=d(o_{\tau},y)$$\approx 0.937$) and $d(x,y)$=$\frac{1}{2\sqrt{2}}$. By placing these values in Equation~\eqref{eqn_1}, we get
    $\angle xo_{\tau}y\approx\cos^{-1}{(0.936)}$ $<\frac{\pi}{4}$ (respectively, $\angle xo_{\tau}y\approx\cos^{-1}{(0.929)}<\frac{\pi}{4}$) for regular 5-gon (respectively,  regular 6-gon).

    Now, consider the case when $\sigma$ is a square or a regular $k$-gon ($k\geq 7$). Note that the length $\ell_\sigma$ of the square tile $s$ is the same for these objects. Among these objects, from the $t_{\sigma}$ value, it is easy to observe that the super-square $S$ is the smallest when $\sigma$ is a square (see Figure~\ref{fig:kgeq7}). So, it is enough to show that $\angle xo_{\tau}y<\frac{\pi}{4}$ when $\sigma$ is a square. When $\sigma$ is a square, we have $d(o_{\tau},x)$~= $d(o_{\tau},y)$~=$\frac{\sqrt{58}}{8}$ and $d(x,y)$=$\frac{1}{\sqrt{2}}$. By placing the values in equation~\eqref{eqn_1}, we get
    $\angle xo_{\tau}y=\cos^{-1}{\left(\frac{21}{29}\right)}<\frac{\pi}{4}$. Thus, the lemma follows.
\end{proof}

\begin{figure}[htbp]
    \centering
    \begin{subfigure}[b]{0.45\textwidth}
        \centering
        \includegraphics[page=3, width=35 mm]{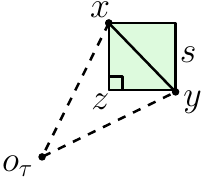}
        \caption{}
        \label{fig:imp}
    \end{subfigure}
    \begin{subfigure}[b]{0.45\textwidth}
        \centering
        \includegraphics[page=4, width=40 mm]{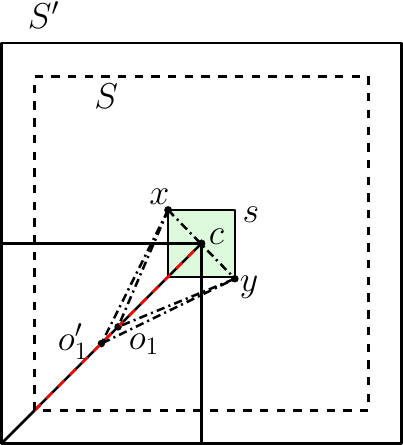}
        \caption{}
        \label{fig:kgeq7}
    \end{subfigure}
    \caption{Description of (a) $\triangle xyz$ and $\triangle xyo_1$ for cone $K_s^{\tau}$, (b) $S$ and $S'$ are super-squares for unit square and unit regular $k$-gon, where $k\geq 7$, respectively.}
    \label{fig:77}
\end{figure}

For $\sigma'\in{\cal D}_s$, let $\tau_{s,{\sigma'}}=\min\{\tau\ |\ o_{\tau}\in \sigma'\}$. Due to Observation~\ref{obs:inter}, the value of $\tau_{s,{\sigma'}}$ is well-defined. Let ${\cal D}_{s,{\tau}}$ denotes the set $\{ \sigma'\in{\cal D}_s\ |\ \tau_{s,{\sigma'}}=\tau\text{ and } \sigma'\cap s\cap\P\neq\emptyset\}$. Let $\sigma'$ and $\sigma''$ be any two translated copies of the regular $k$-gon $\sigma$. So, note that the intersection $\partial \sigma'\cap\partial \sigma''$ is either empty or consists of two connected components.  For example, see Figure~\ref{fig:case_6}. The following lemma says that if we consider $\sigma',\sigma'' \in {\cal D}_{s,\tau}$ and $\partial \sigma'\cap\partial \sigma''$ is non-empty, then at most one of the connected components can intersect the subregion $K_s^{\tau}$ which we will prove by using Lemma~\ref{lemma:regular} and Lemma~\ref{lemma:cone}.

\begin{lemma}\label{lemma:psLine_k_s}
    For any pair of objects $\sigma',\sigma'' \in {\cal D}_{s,\tau}$, the intersection $\partial\sigma'\cap \partial\sigma''\cap K_s^{\tau}$ consists of at most one connected component of $\partial\sigma'\cap\partial \sigma''$.
\end{lemma}
\begin{proof}
    For a contradiction, let us assume that $\partial \sigma'\cap\partial \sigma''\cap K_s^{\tau}$ consists of two connected components. Let $p$ and $q$ be two points lying in the two connected components. It is easy to observe that $\sigma'\cap \sigma''\cap K_s^\tau$ contains the triangle $\triangle po_{\tau} q$. Now, due to Lemma~\ref{lemma:cone}, we know that the opening angle of the cone $K_s^{\tau}$ is strictly less than $\frac{\pi}{4}$. This leads to a contradiction to Lemma~\ref{lemma:regular}.
\end{proof}

For every tile $s$ from the partition $P_{\sigma}$ and every $\tau\in[4]$, we define a set $V_{s,\tau}$ of extreme points as follows.
$V_{s,{\tau}}=\{p\in\P\ |\ \exists\  \sigma'\in {\cal D}_{s,{\tau}}\text{ such that } p\in\partial\sigma' \text{ and } \sigma'\cap s\cap\P=\{p\}\}$.
Note that if $ \sigma'\in {\cal D}_{s,{\tau}}$ and $\sigma'\cap s\cap\P\neq\emptyset$, then $\sigma'\cap V_{s,{\tau}}\neq\emptyset$.

\vspace{1mm}\noindent\hspace{2mm} $\bullet$ \underline{Vertex Ranking of Extreme Points}: 
Let $\theta_{s,{\tau}}:V_{s,{\tau}}\rightarrow \left[0,2\pi\right]$ denote an angle function, where $\theta_{s,{\tau}}(p)$ is equal to the slope of the line $\overline{o_{\tau} p}$. It is easy to see that $\theta_{s,{\tau}}$ is a one-to-one function. Let $\{q_i\}_{i=1}^{|V_{s,{\tau}}|}$ denote an ordering of $V_{s,{\tau}}$ in increasing order of the angle function $\theta_{s,{\tau}}$. For any object $\sigma'\in{\cal D}_{s,{\tau}}$, the intersection $\sigma'\cap V_{s,{\tau}}$ is an \emph{interval} if there exist $i, k$ such that $\sigma'\cap V_{s,{\tau}}=\{q_j\ |\ i\leq j \leq k\}$.
Now, the following lemma plays an important role in proving the main result which we will prove by using Lemma~\ref{lemma:psLine_k_s}.

\begin{lemma}\label{lem:interval}
    For every $\sigma'\in{\cal D}_{s,{\tau}}$, the intersection $\sigma'\cap V_{s,{\tau}}$ is an interval.
\end{lemma}

\begin{proof}
    For a contradiction, let us assume that there exists a $\sigma'\in {\cal D}_{s,{\tau}}$ containing extreme points $q_j$ and $q_l$ but not $q_k$, where $\theta_{s,\tau}(q_j) < \theta_{s,\tau}(q_k) < \theta_{s,\tau}(q_l)$ (see Figure~\ref{fig:prop_int}). 
    Since $q_k$ is an extreme point, there exists a regular $k$-gon $\sigma''\in{\cal D}_{s,{\tau}}$ containing $q_k$ but neither $q_j$ nor $q_l$. It implies that $\partial \sigma'$ and $\partial \sigma''$ intersect twice inside the tile $s$. Consequently, $\sigma'\cap \sigma''\cap K_s^i$ consists of two connected component of $\partial\sigma'\cap\partial \sigma''$, which contradicts Lemma~\ref{lemma:psLine_k_s}. 
\end{proof}

Let $G_{s,\tau}$ denotes the path graph over the extreme points $V_{s,{\tau}}$, where $q_i$ is a neighbor of $q_{i+1}$, for any $i\in[|V_{s,{\tau}}|-1]$. 
Due to~\cite[Proposition~5]{EvenS14}, we have that the vertex ranking number of a path graph having $n$ vertices is $\lfloor\log_2 2n\rfloor$.
Let $c_{s,\tau}:V_{s,{\tau}} \rightarrow [k]$ denotes a vertex ranking of the path graph $G_{s,\tau}$ that uses $\lfloor \log_2(2|V_{s,{\tau}}|)\rfloor$ colors. 

\begin{figure}[htbp]
    \centering
    \includegraphics[page=5, width=30 mm]{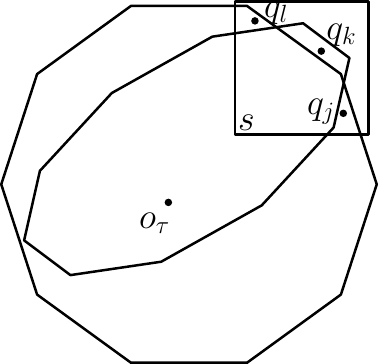}
    \caption{Illustration of Lemma~\ref{lem:interval}.}
    \label{fig:prop_int}
\end{figure}

\vspace{2mm}
Now, we describe the algorithm $\ES$.

\vspace{1mm}\noindent\textbf{Description of $\ES$.}
The algorithm maintains a hitting set $\H_{i-1}$ for all the $i-1$ objects $\{\sigma_1,\sigma_2,\ldots, \sigma_{i-1}\}$ that have been the part of input so far. Upon the arrival of a new object $\sigma_i$, if it is already stabbed  by $\H_{i-1}$, then simply update $\H_i\leftarrow \H_{i-1}$. Otherwise, from each tile $s$ such that $\sigma_i\cap\P\cap s\neq\emptyset$, first determine the $\tau$ such that $\tau=\tau_{s,{\sigma_i}}$ and then select the point $v_{s,\sigma_i}$ such that $c_{s,\tau} (v_{s,\sigma_i})= \max\{c_{s,\tau}(v)\ |\ v \in \sigma_i\cap\P \cap V_{s,{\tau}}\}$. These points are added to $\H_{i-1}$ to obtain $\H_i$.

\vspace{2mm}
In the following lemma, it is demonstrated that upon the arrival of two unstabbed translated copies of $\sigma$ that have a common point $p\in s$ and are of the same type $\tau$, then they are stabbed by extreme points in $V_{s,{\tau}}$ of different colors. We prove this by using the Lemma~\ref{lem:interval}.

\vspace{1mm}\noindent
$\blacktriangleright$ \textbf{Lemma~\ref{lem:distinct}.}\label{App_lem:distinct}
\textit{For any $\tau\in[4]$, if  $\sigma_i,\sigma_j\in{\cal D}_{s,{\tau}}$ are unstabbed upon their arrivals, and there exists a point $p\in \sigma_i\cap \sigma_j\cap\P \cap s$, then $c_{s,\tau}(v_{s,\sigma_i})\neq c_{s,\tau}(v_{s,\sigma_j})$.}
\begin{proof}
    Assume for a contradiction that $c_{s,\tau}(v_{s,\sigma_i})=c_{s,\tau}(v_{s,\sigma_j})$. According to Lemma~\ref{lem:interval}, the intersection $\sigma_i\cap V_{s,{\tau}}$ is an interval, i.e., $\sigma_i\cap V_{s,{\tau}}=I_i=[a_i, b_i]$. Likewise, $I_j=[a_j, b_j]$ represents the interval for $\sigma_j\cap V_{s,{\tau}}$. Given that $I_i=[a_i, b_i]$ and $I_j=[a_j, b_j]$ are both intervals having a common point, thus $I_i\cup I_j$ is also an interval. 
    According to the definition of vertex ranking, $I_i\cup I_j$ contains a vertex $k$ having color $c(k)$ such that $c(k)>c_{s,\tau}(v_{s,\sigma_i})$. Now, either $I_i$ or $I_j$ contains the vertex $k$. W.l.o.g., let us assume that $I_i$ contains the vertex $k$ having color $c(k)$ such that $c(k)>c_{s,\tau}(v_{s,\sigma_i})$. But according to our algorithm, $c_{s,\tau}(v_{s,\sigma_i})$ is the largest color in $I_i$. This leads to a contradiction.
\end{proof}

\subsection{Missing Proof of Lemma~\ref{lem:shrunk}}

\noindent
$\blacktriangleright$ \textbf{Lemma~\ref{lem:shrunk}.}\label{App_lem:shrunk}
\textit{$\sigma=\cup_{\sigma'\in{\cal K}_{\sigma}} \sigma'$.
}

\begin{figure}[htbp]
\centering
\begin{subfigure}[b]{0.30\textwidth}
    \centering
    \includegraphics[page=1, width= 35 mm]{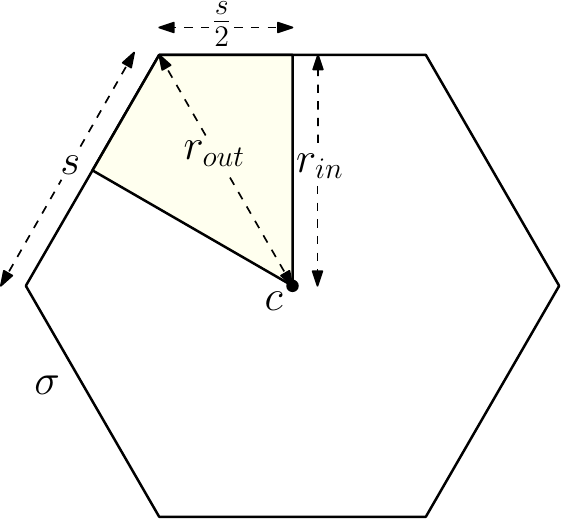}
    \caption{}
    \label{fig:clm_1e}
\end{subfigure}
\centering
\begin{subfigure}[b]{0.30\textwidth}
    \centering
    \includegraphics[page=2, width=35 mm]{Shrunk.pdf}
    \caption{}
    \label{fig:clm_2e}
\end{subfigure}
\centering
\begin{subfigure}[b]{0.34\textwidth}
    \centering
    \includegraphics[page=3, width=35 mm]{Shrunk.pdf}
    \caption{}
    \label{fig:clm_3e}
\end{subfigure}
\caption{Illustration of Case~1. (a) Illustration of the partitioning of $\sigma$ using $k$ congruent kites whose two distinct side lengths are $\frac{s}{2}$ and $w$. (b) Detailed description the cell $C_i$, $\sigma_i'$ and $\sigma_0'$. (c) Description of $c'$ and $a$, the two opposite vertices of $\sigma_i'$}
\label{fig:clm_shrunk_e}
\end{figure}

\begin{proof}
    As per the definition of the shrunk set $K_{\sigma}$, it is easy to observe that  $\cup_{\sigma'\in K_{\sigma}}\sigma'\subseteq \sigma$.
    To prove the lemma, we will show that $\sigma\subseteq\cup_{\sigma'\in K_{\sigma}}\sigma'$. Let $w$ be the width of $\sigma$. When $w$ is an integral power of $2$, then $K_{\sigma}=\{\sigma\}$, and we have nothing to prove. When $w\in(2^i,2^{i+1})$, then as per construction of ${\cal K}_{\sigma}$, it contains $k+1$ copies 
    of $\sigma_0'$ i.e., $\sigma_1',\sigma_2',\ldots,\sigma_k'$ along with $\sigma_0'$. Let $s$ and $s'$  be the side length of $\sigma$ and $\sigma_0'$, respectively.
    Let $P$ be a partition of $\sigma$ using $k$ congruent kites (a quadrilateral with two pairs of adjacent sides equal in length) whose two distinct side lengths are $\frac{s}{2}$ and $w$, respectively (for an illustration, see Figures~\ref{fig:clm_1} and~\ref{fig:clm_2}).
    Each kite of $P$ is known as a \emph{cell}. For each $i\in[k]$, let $C_i=\square abcd$ be a kite, where the point $a$ coincides with the corner $c_i$ of the regular $k$-gon and $c$ is the center of $\sigma$ (see Figure~\ref{fig:clm_3}). To prove the lemma, it is sufficient to prove that for each $i\in[k]$, we can cover the cell $C_i$ using $\sigma_0'$ and $\sigma'_i$ from the shrunk set $K_{\sigma}$.
    Let $r_{in}$ (respectively, $r_{in}'$) be the radius of the largest ball enclosed in $\sigma$ (respectively, $\sigma_0'$), whereas $r_{out}$ (respectively, $r_{out}'$)  be the radius of the smallest ball enclosing $\sigma$ (respectively, $\sigma_0'$).
    Notice that $d(a,b)=d(a,d)=\frac{s}{2}\leq s'$, $d(c,b)=d(c,d)=\frac{r_{in}}{2}\leq r_{in}'$ and $d(a,c)=\frac{r_{out}}{2}\leq r_{out}'$.
    Now, depending on whether $k$ is even or odd, we have the following cases.

    \noindent\textbf{Case 1.} $k$ is even. Let $c'$ and $a$ are the two opposite vertices of $\sigma_i'$ (see Figure~\ref{fig:clm_shrunk_e}). Since $d(a,c')=2r_{out}'\geq r_{out}$ and $d(a,c)=r_{out}$, $\sigma_i'$ contains $a$ and $c$ . Since the angle $\angle bad$ is equal to the interior angle of regular $k$-gon $\sigma_i'$ (or $\sigma$) and $s'\geq d(a,b)\left(=\frac{s}{2}\right)$, all points $a,b,c$ and $d$ lie in $\sigma_i'$. Hence, $\sigma_i'$ contains the cell $C_i=\square abcd$.

    \noindent\textbf{Case 2.} $k$ is odd. For this case, the regular $k$-gon $\sigma_i'$ contains the point $a$ (common vertex of $\sigma_i'$ and $\sigma$) but may or may not contain the point $c$.
    If $\sigma_i'$ contains both $a$ and $c$, then similar to the previous case, $\sigma_i'$ contains the cell $C_i=\square abcd$.
    When $\sigma_i'$ does not contain $c$, then $\sigma_i'$ does not fully contain the cell $C_i$. Notice that for a regular odd $k$-gon, there is an edge opposite to a vertex. Let $e$ be an edge opposite to the vertex $a$. Let $L$ be a line obtained by extending $e$ in both directions.  Since $d(a,L)\geq r_{out}'+r_{in}'\geq \frac{r_{out}+r_{in}}{2}$, the line $L$ intersects $C_i$ on the edge $bc$ at $b'$ and edge $dc$ at $d'$. 
    Now, first, we will show that $\sigma_i'$ contains the polygon $abb'd'da$. Then, we will prove that $\sigma_0'$ contains the triangle $\triangle b'cd'$.
    Let $\overline{ac}$ be a line segment joining points $a$ and $c$, whereas $\overline{b'd'}$ be a line segment joining points $b'$ and $d'$. Notice that $\overline{ac}$ partitions $\overline{b'd'}$ and $e$ into two equal parts. The line segments $\overline{ac}$ and $\overline{b'd'}$ intersect at a point. Let $f$ be that point. Before proving that $\sigma_i'$ contains the polygon $abb'd'da$, we prove the following claim in which we show that $d(b',f)\leq \frac{s'}{2}$.

\begin{figure}[htbp]
\centering
\begin{subfigure}[b]{0.30\textwidth}
    \centering
    \includegraphics[page=4, width= 35 mm]{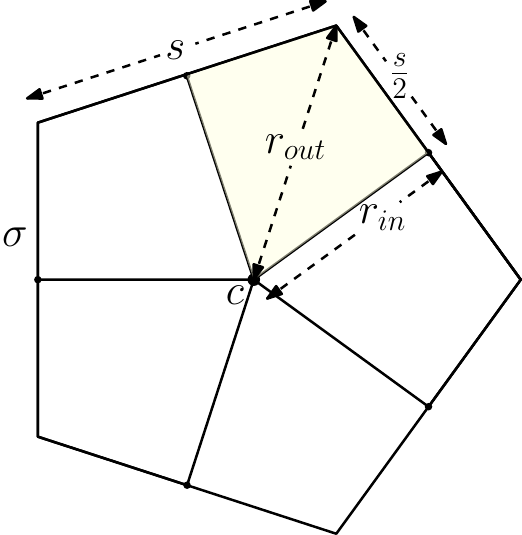}
    \caption{}
    \label{fig:clm_1}
\end{subfigure}
\centering
\begin{subfigure}[b]{0.30\textwidth}
    \centering
    \includegraphics[page=5, width=35 mm]{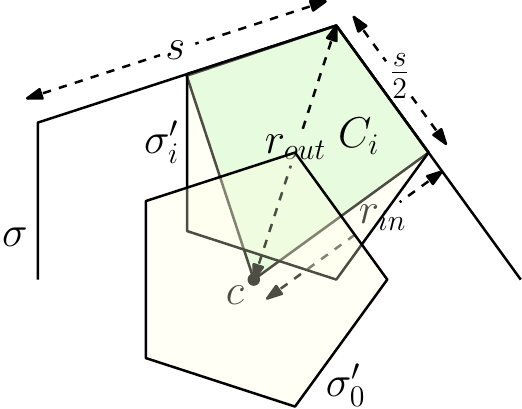}
    \caption{}
    \label{fig:clm_2}
\end{subfigure}
\begin{subfigure}[b]{0.34\textwidth}
    \centering
    \includegraphics[page=6, width=35 mm]{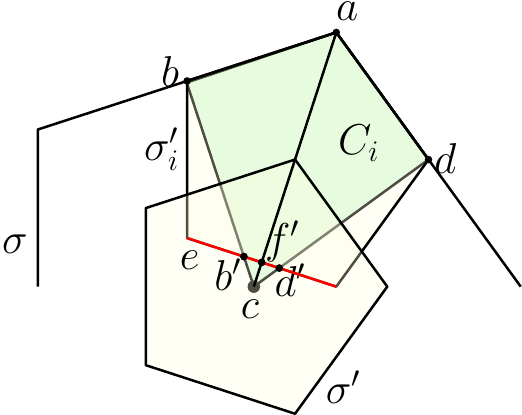}
    \caption{}
    \label{fig:clm_3}
\end{subfigure}
\caption{Illustration of Case~2. (a) Illustration of the partitioning of $\sigma$ using $k$ congruent kites whose two distinct side lengths are $\frac{s}{2}$ and $w$. (b) Detailed description the cell $C_i$, $\sigma_i'$ and $\sigma_0'$. (c) Description of $\triangle b'fc$. The red-colored line segment is an edge $e$, opposite to the vertex $a$, of $\sigma_i'$.}
\label{fig:clm_shrunk}
\end{figure}

\begin{claim}\label{clm:dist}
    $d(b',f)\leq \frac{s'}{2}$
\end{claim}

\begin{proof}
    Observe Figure~\ref{fig:clm_3} and consider the right angle triangle $\triangle b'fc$ right angled at $\angle b'fc$. Since we partition $\sigma$ using $k$ congruent kites, the angle $\angle b'cd'=\frac{2\pi}{k}$.  As a result, the angle $\angle b'cf=\frac{\pi}{k}$. Thus, due to the angle sum property of a triangle, we have $\angle cb'f=\left(\frac{\pi}{2}-\frac{\pi}{k}\right)$.

    Now, using the law of sines, we have the following
    
   \begin{align*}
       d(f,b')&=d(f,c)\frac{\sin \left(\frac{\pi}{k}\right)}{\sin \left(\frac{\pi}{2}-\frac{\pi}{k}\right)}\leq d(f,c)\frac{\sin \left(\frac{\pi}{k}\right)}{\cos \left(\frac{\pi}{k}\right)}\tag{$\sin \left(\frac{\pi}{2}-\theta\right)=\cos (\theta)$}\\
        &\leq (r_{out}-r_{out}'-r_{in}')\frac{\sin \left(\frac{\pi}{k}\right)}{\cos \left(\frac{\pi}{k}\right)}\tag{$d(f,c)=r_{out}-r_{out}'-r_{in}'$}\\
        &\leq \left(r_{out}-\frac{r_{out}}{2}-\frac{r_{in}}{2}\right)\frac{\sin \left(\frac{\pi}{k}\right)}{\cos \left(\frac{\pi}{k}\right)}\tag{$r_{out}'\geq \frac{r_{out}}{2}$ and $r_{in}'\geq \frac{r_{in}}{2}$}\\
        &\leq \left(\frac{r_{out}}{2}-\frac{r_{out}\cos\left(\frac{\pi}{k}\right)}{2}\right)\frac{\sin \left(\frac{\pi}{k}\right)}{\cos \left(\frac{\pi}{k}\right)}\tag{$r_{in}= r_{out}\cos\left(\frac{\pi}{k}\right)$}\\
        &\leq \frac{r_{out}}{2}\left(\frac{1}{\cos \left(\frac{\pi}{k}\right)}-1\right)\sin \left(\frac{\pi}{k}\right)\\
         &< \frac{r_{out}}{2}\left(\sqrt{2}-1\right)\sin \left(\frac{\pi}{k}\right)<\frac{r_{out}}{2}\sin \left(\frac{\pi}{k}\right)\tag{$\frac{1}{cos\left(\frac{\pi}{k}\right)}\leq \sqrt{2}$, for $k\geq 4$}\\
          &< \frac{s}{4}\tag{Since $s=2r_{out}\sin\left(\frac{\pi}{k}\right)$ }\\
          &< \frac{s'}{2} \tag{Since $s'\geq \frac{s}{2}$}.
   \end{align*}
\end{proof}

Due to Claim~\ref{clm:dist}, we have $d(b',d')=2d(f,b')<2\frac{s'}{2}=s'$. Thus, $e$ contains $\overline{b'd'}$. As a result, all points $a,b,b',d'$ and $d$ lie on the boundary of $\sigma_i'$. Thus, $\sigma_i'$ contains the polygon $abb'd'da$. Now, we only need to show that $\sigma_0'$ contains $\triangle b'cd'$. Before proving this, we prove the following claim.

\begin{claim}\label{clm:tri}
    $d(b',c)\leq r_{in}'$.    
\end{claim}

\begin{proof}
    See Figure~\ref{fig:clm_3} and consider the right angle $\angle b'fc$. Recall that the angle $\angle b'cf=\frac{\pi}{k}$, $\angle b'fc=\frac{\pi}{2}$ and $\angle cb'f=\left(\frac{\pi}{2}-\frac{\pi}{k}\right)$.

    Now, using the law of sines, we have the following
    
   \begin{align*}
       d(b',c)&=d(f,c)\frac{\sin \left(\frac{\pi}{2}\right)}{\sin \left(\frac{\pi}{2}-\frac{\pi}{k}\right)}\leq d(f,c)\frac{1}{\cos \left(\frac{\pi}{k}\right)}\tag{$\sin\left(\frac{\pi}{2}\right)=1$ and $\sin \left(\frac{\pi}{2}-\theta\right)=\cos (\theta)$}\\
        &\leq \frac{(r_{out}-r_{out}'-r_{in}')}{\cos \left(\frac{\pi}{k}\right)}\tag{$d(f,c)=r_{out}-r_{out}'-r_{in}'$}\\
        &\leq \frac{1}{\cos \left(\frac{\pi}{k}\right)}{\left(r_{out}-\frac{r_{out}}{2}-\frac{r_{in}}{2}\right)}=\frac{1}{\cos \left(\frac{\pi}{k}\right)}{\left(\frac{r_{out}}{2}-\frac{r_{in}}{2}\right)}\tag{$r_{out}'\geq \frac{r_{out}}{2}$ and $r_{in}'\geq \frac{r_{in}}{2}$}\\
        &\leq \frac{1}{\cos \left(\frac{\pi}{k}\right)}{\left(\frac{r_{in}}{2\cos \left(\frac{\pi}{k}\right)}-\frac{r_{in}}{2}\right)}\tag{$r_{in}= r_{out}\cos\left(\frac{\pi}{k}\right)$}\\
         &\leq \sqrt{2}{\left(\frac{r_{in}\sqrt{2}}{2}-\frac{r_{in}}{2}\right)}\tag{$\frac{1}{cos\left(\frac{\pi}{k}\right)}\leq \sqrt{2}$, for $k\geq 4$}\\
         &\leq \frac{r_{in}}{2}{\left(2-\sqrt{2}\right)}< \frac{r_{in}}{2}\tag{Since $r_{in}'\geq \frac{r_{in}}{2}$}\\
         &< r_{in}'\tag{Since $r_{in}'\geq \frac{r_{in}}{2}$}.
   \end{align*}
\end{proof}

    Due to Claim~\ref{clm:tri}, we have $d(b',c)\leq r_{in}'$. As a result, all points $b',c$ and $d'$ lie in $\sigma_0'$. Thus, $\sigma_0'$ contains $\triangle b'cd'$. Notice that $\cup_{i\in[k]} C_i=\sigma$ and $C_i\subseteq\sigma_i'\cup\sigma_0'$. Thus, $\sigma\subseteq\cup_{\sigma_0'\in K_{\sigma}}\sigma_0'$. Hence, the lemma follows.
\end{proof}

\end{document}